%% file: sublinear_lra.tex
\documentclass[11pt]{article}
\usepackage{fullpage}
\usepackage{amsmath,amsfonts,amsthm,amssymb}
\usepackage{url}

\usepackage{color}
\usepackage[usenames,dvipsnames,svgnames,table]{xcolor}
\usepackage[colorlinks=true, linkcolor=red, urlcolor=blue, citecolor=blue]{hyperref}

\usepackage{booktabs}
\usepackage{algorithm}
\usepackage{algpseudocode}
\usepackage{comment}
\usepackage{mathtools}
\usepackage{microtype}
\usepackage{caption}
\usepackage{float} 
\usepackage{graphicx}
\usepackage{subcaption}
\usepackage{listings}

\usepackage[framemethod=TikZ]{mdframed}
\newcounter{Frame}
\newenvironment{Frame}[1][htb]{%
\refstepcounter{Frame}
    \begin{mdframed}[%
        frametitle={#1},
        skipabove=\baselineskip plus 2pt minus 1pt,
        skipbelow=\baselineskip plus 2pt minus 1pt,
        linewidth=1.0pt,
        frametitlerule=true,
    ]%
}{%
    \end{mdframed}
}

\numberwithin{equation}{section}
\numberwithin{figure}{section}
\theoremstyle{plain}
	\newtheorem{theorem}{Theorem}[section]

	\newtheorem{lemma}[theorem]{Lemma}

	\newtheorem{corollary}[theorem]{Corollary}

\theoremstyle{definition}
	\newtheorem{definition}[theorem]{Definition}
	
	\newtheorem*{remark*}{Remark}

\def\prob#1#2{\mbox{\bf Pr}_{#1}\left[ #2 \right]}

\def\expec#1#2{{\bf \mathbb{E}}_{#1}[ #2 ]}
\def\expecf#1#2{{\bf \mathbb{E}}_{#1}\left[ #2 \right]}
\def\var#1{\mbox{\bf Var}[ #1 ]}
\def\varf#1{\mbox{\bf Var}\left[ #1 \right]}

\def\trace#1{\mathrm{Tr} \left(#1 \right)}

\global\long\def\argmin{\mathrm{argmin}}
\global\long\def\argmax{\mathrm{argmax}}

\def\AA{\mathbf{A}}
\def\BB{\mathbf{B}}
\def\MM{\mathbf{M}}
\def\NN{\mathbf{N}}
\def\LL{\mathbf{L}}
\def\WW{\mathbf{W}}
\def\CC{\mathbf{C}}
\def\D{\mathbf{D}}
\def\II{\mathbb{I}}
\def\Y{\mathbf{Y}}
\def\X{\mathbf{X}}
\def\P{\mathbf{P}}
\def\U{\mathbf{U}}
\def\V{\mathbf{V}}
\def\E{\mathbf{E}}
\def\S{\mathbf{S}}
\def\T{\mathbf{T}}

\def\PP{\mathcal{P}}
\def\QQ{\mathcal{Q}}

\title{Sublinear Time Low-Rank Approximation of Distance Matrices}
\usepackage{times}
\author{Ainesh Bakshi\thanks{The authors 
thank the partial support by the 
National Science Foundation under Grant
No. CCF-1815840.
}\\
Carnegie Mellon University\\
abakshi@cs.cmu.edu
\and David P. Woodruff\footnotemark[1]
\\
Carnegie Mellon University\\
dwoodruf@cs.cmu.edu}

\begin{document}

\maketitle
\input{abstract.tex}
\newpage
\input{intro.tex}

\input{overview.tex}
\input{first_sublinear.tex}
\input{full_algorithm.tex}
\input{relative_error.tex}

\bibliographystyle{alpha}
\bibliography{colt-2018}

\newpage
\input{appendix.tex}

\end{document}

%% file: abstract.tex
\begin{abstract}
Let $\PP=\{ p_1, p_2, \ldots p_n \}$ and $\QQ = \{ q_1, q_2 \ldots q_m \}$ be two point sets in an arbitrary metric space. Let $\AA$ represent the $m\times n$ pairwise distance matrix with $\AA_{i,j} = d(p_i, q_j)$. Such distance matrices are commonly computed in software packages and have applications to learning image manifolds, handwriting recognition, and multi-dimensional unfolding, among other things. In an attempt to reduce their description size, we study low rank approximation of such matrices. Our main result is to show that for any underlying distance metric $d$, it is possible to achieve an additive error low-rank approximation in sublinear time. We note that it is provably impossible to achieve such a guarantee in sublinear time for arbitrary matrices $\AA$, and consequently our proof exploits special properties of distance matrices. We develop a recursive algorithm based on additive projection-cost preserving sampling. We then show that in general, relative error approximation in sublinear time is impossible for distance matrices, even if one allows for bicriteria solutions. Additionally, we show that if $\PP = \QQ$ and $d$ is the squared Euclidean distance, which is not a metric but rather the square of a metric, then a relative error bicriteria solution can be found in sublinear time.
\end{abstract}

%% file: intro.tex
\section{Introduction}
\label{introduction}
We study low rank approximation of matrices $\AA$ formed by the pairwise distances
between two (possibly equal) sets of points or observations $\PP = \{p_1, \ldots, p_m\}$ and $\QQ = \{q_1, \ldots, q_n\}$ in an arbitrary underlying metric space. That is, 
$\AA$ is an $m \times n$ matrix for which $A_{i,j} = d(p_i, q_j)$. Such distance
matrices are the outputs of routines in commonly used software packages
such as the pairwise command in Julia, the pdist2 command in Matlab, or the crossdist command in R. 

Distance matrices have found many applications in machine learning, where Weinberger and Sauk use them to learn image manifolds \cite{ws04}, Tenenbaum, De Silva, and Langford use them for image understanding and handwriting recognition \cite{t00}, Jain and Saul use them for speech and music \cite{js04}, and Demaine et al. use them for music and musical rhythms \cite{DHMRTTWW09}. For an excellent tutorial on Euclidean distance matrices, we refer the reader to \cite{d15}, which lists applications to nuclear magnetic resonance (NMR), crystallagraphy, visualizing protein structure, and multi-dimensional unfolding. 

We consider the most general case for which $\PP$ and $\QQ$ are not necessarily the same set of points. For example, one may have two large unordered sets of samples from some distribution, and may want to determine how similar (or dissimilar) the sample sets are to each other. Such problems arise in hierarchical clustering and phylogenetic analysis\footnote{See, e.g., \url{https://en.wikipedia.org/wiki/Distance_matrix}}. 

Formally, let $\PP=\{ p_1, p_2, \ldots p_m \}$ and $\QQ = \{ q_1, q_2 \ldots q_n \}$ be two sets of points in an arbitrary metric space.
Let $\AA$ represent the $m\times n$ pairwise distance matrix with $A_{i,j} = d(p_i, q_j)$. Since the matrix $\AA$ may be very large, it is often desirable to reduce the number of parameters needed to describe it. Two standard methods of doing this are via sparsity and low rank approximation. In the distance matrix setting, if one first filters $\PP$ and $\QQ$ to contain only distinct points, then each row and column can contain at most a single zero entry, so typically such  matrices $\AA$ are dense. Low-rank approximation, on the other hand, can be highly beneficial since if the point sets can be clustered into a small number of clusters, then each cluster can be used to define an approximately rank-$1$ component, and so $\AA$ is an approximately low rank matrix. 

To find a low-rank factorization of $\AA$, one can compute its singular value decomposition (SVD), though in practice this takes $\min(mn^2, m^2n)$ time. One can do slightly better with theoretical algorithms for fast matrix multiplication, though not only are they impractical, but there exist much faster randomized approximation algorithms. Indeed, one can use Fast Johnson Lindenstrauss transforms (FJLT) \cite{sarlos2006improved}, or CountSketch matrices \cite{clarkson2013low,mm13,NN13,BDN15,c16}, which for dense matrices $\AA$, run in $O(mn) + (m+n)\textrm{poly}(k/\epsilon)$ time.

At first glance the $O(mn)$ time seems like it could be optimal. Indeed, for arbitrary $m \times n$ matrices $\AA$, outputting a rank-$k$ matrix $\BB$ for which 
\begin{equation}
\label{eqn:additive}
    \|\AA-\BB\|_F^2 \leq \|\AA-\AA_k\|_F^2 + \epsilon \|\AA\|_F^2
\end{equation}
can be shown to require $\Omega(mn)$ time. Here $\AA_k$ denotes the best rank-$k$ approximation to $\AA$ in Frobenius norm, and recall for an $m \times n$ matrix $\CC$, $\|\CC\|_F^2 = \sum_{i = 1, \ldots, m, j = 1, \ldots, n}C_{i,j}^2$. The \textit{additive error} guarantee above is common in low-rank approximation literature and appears in \cite{fkv04}. To see this lower bound, note that if one does not read nearly all the entries of $\AA$, then with good probability one may miss an entry of $\AA$ which is arbitrarily large, and therefore cannot achieve (\ref{eqn:additive}). 

Perhaps surprisingly, \cite{mw17} show that for positive semidefinite (PSD) $n \times n$ matrices $\AA$, one can achieve (\ref{eqn:additive}) in {\it sublinear} time, namely, in $n \cdot k \cdot \textrm{poly}(1/\epsilon)$ time. Moreover, they achieve the stronger notion of relative error, that is, they output a rank-$k$ matrix $\BB$ for which 
\begin{eqnarray}\label{eqn:relative}
\|\AA-\BB\|_F^2 \leq (1+\epsilon)\|\AA-\AA_k\|_F^2. 
\end{eqnarray}
The intuition behind their result is that the ``large entries'' causing the $\Omega(mn)$ lower bound cannot hide in a PSD matrix, since they necessarily create large diagonal entries. 

A natural question is whether it is possible to obtain low-rank approximation algorithms for distance matrices in sublinear time as well. A driving intuition that it may be possible is no matter which metric the underlying points reside in, they necessarily satisfy the triangle inequality. Therefore, if $\AA_{i,j} = d(p_i, q_j)$ is large, then since $d(p_i, q_j) \leq d(p_i, q_1) + d(q_1, p_1) + d(p_1, q_j)$, at least one of $d(p_i, q_1), d(q_1, p_1), d(p_1, q_j)$ is large, and further, all these distances can be found by reading the first row and column of $\AA$. Thus, large entries cannot hide in the matrix. 

Are there sublinear time algorithms achieving (\ref{eqn:additive})? Are there sublinear time algorithms achieving (\ref{eqn:relative})? These are the questions we put forth and study in this paper.

\subsection{Our Results}
Our main result is that we obtain sublinear time algorithms achieving the additive error guarantee similar to (\ref{eqn:additive}) for distance matrices, which is impossible for general matrices $\AA$. We show that for {\it every metric} $d$, this is indeed possible. Namely, for an arbitrarily small constant $\gamma > 0$, we give an algorithm running in $\widetilde{O}((m^{1+\gamma}+ n^{1+\gamma}) \textrm{poly}(k \epsilon^{-1}))$ time and achieving guarantee $\| \AA - \MM \NN^T \|^2_{F} \leq  \| \AA - \AA_k\|^2_{F} + \epsilon \| \AA \|^2_F$, for any distance matrix with metric $d$. Note that our running time is significantly {\it sublinear} in the description size of $\AA$. Indeed, thinking of the shortest path metric on an unweighted bipartite graph in which $\PP$ corresponds to the left set of vertices, and $\QQ$ corresponds to the right set of vertices, for each pair of points $p_i \in \PP$ and $q_j \in \QQ$, one can choose $d(p_i, q_j) = 1$ or $d(p_i, q_j) > 1$ independently of all other distances by deciding whether to include the edge $\{p_i, q_j\}$. Consequently, there are at least $2^{\Omega(mn)}$ possible distance matrices $\AA$, and since our algorithm reads $o(mn)$ entries of $\AA$, cannot learn whether $d(p_i, q_j) = 1$ or $d(p_i, q_j) > 1$ for each $i$ and $j$. Nevertheless, it still learns enough information to compute a low rank approximation to $\AA$. 

We note that a near matching lower bound holds just to write down the output of a factorization of a rank-$k$ matrix $\BB$ into an $m \times k$ and a $k \times n$ matrix. Thus, up to an $(m^{\gamma} + n^{\gamma}) \textrm{poly}(k \epsilon^{-1})$ factor, our algorithm is also optimal among those achieving the additive error guarantee of (\ref{eqn:additive}).

A natural followup question is to consider achieving relative error (\ref{eqn:relative}) in sublinear time. Although large entries in a distance matrix $\AA$ cannot hide, we show it is still impossible to achieve the relative error guarantee in less than $mn$ time for distance matrices. That is, we show for the $\ell_{\infty}$ distance metric, that there are instances of distance matrices $\AA$ with unequal $\PP$ and $\QQ$ for which even for $k = 2$ and any constant accuracy $\epsilon$, must read $\Omega(mn)$ entries of $\AA$. In fact, our lower bound holds even if the algorithm is allowed to output a rank-$k'$ approximation for any $2 \leq k' = o(\min(m,n))$ whose cost is at most that of the best rank-$2$ approximation to $\AA$. We call the latter a bicriteria algorithm, since its output rank $k'$ may be larger than the desired rank $k$. Therefore, in some sense obtaining additive error (\ref{eqn:additive}) is the best we can hope for. 

We next consider the important class of Euclidean matrices for which the entries correspond to the {\it square} of the Euclidean distance, and for which $\PP = \QQ$. In this case, we are able to show that if we allow the low rank matrix $\BB$ output to be of rank $k+4$, then one can achieve the relative error guarantee of (\ref{eqn:relative}) with respect to the best rank-$k$ approximation, namely, that $$\|\AA-\BB\|_F^2 \leq (1+\epsilon)\|\AA-\AA_k\|_F^2.$$ Further, our algorithm runs in a sublinear $n \cdot k \cdot \textrm{poly}(1/\epsilon)$ amount of time. 
Thus, our lower bound ruling out sublinear time algorithms achieving (\ref{eqn:relative}) for bicriteria algorithms cannot hold for this class of matrices.

%% file: overview.tex
\section{Technical Overview}
Given a $m \times n$ matrix $\AA$ with rank $r$, we can compute its 
singular value decomposition, denoted by $\texttt{SVD}(\AA) = \U 
\mathbf{\Sigma} \V^T$, such that $\U$ is a $m \times r$ matrix with 
orthonormal columns, $\V^T$ is a $r \times n$ matrix with orthonormal 
rows and $\mathbf{\Sigma}$ is a $r \times r$ diagonal matrix. The entries 
along the diagonal are the singular values of $\AA$, denoted by 
$\sigma_1, \sigma_2 \ldots \sigma_r$. Given an integer $k \leq r$, we 
define the truncated singular value decomposition of $\AA$ that zeros out 
all but the top $k$ singular values of $\AA$, i.e. $\AA_k = \U 
\mathbf{\Sigma}_k \V^T$, where $\mathbf{\Sigma}_k$ has only $k$ non-zero 
entries along the diagonal. It is well known that the truncated svd 
computes the best rank-$k$ approximation to $\AA$ under the frobenius 
norm, i.e. $\AA_k = \min_{\textrm{rank}(\X)\leq k} \| \AA - \X\|_F$.
More generally, for any matrix $\MM$, we use the notation $\MM_k$ and 
$\MM_{\setminus k}$ to denote the first $k$ components and all but the 
first $k$ components respectively. We use $\MM_{i,*}$ and $\MM_{*,j}$ to 
refer to the $i^{th}$ row and $j^{th}$ column of $\MM$ respectively.  

Our starting point is a result of Frieze, Kannan and Vempala 
\cite{fkv04}, that states sampling columns weighted by their column norms 
and then sub-sampling rows by their row norms suffices to obtain additive 
error low-rank approximation guarantees.  However, if the input matrix 
has no structure, it is impossible to obtain estimates to the column 
norms in sublinear time. Restricting to the family of distance matrices, 
we show that we can obtain coarse estimates to the column norms by 
uniformly sampling entries in each column and constructing a biased 
estimator. 

In particular, we show that if the input matrix $\AA$ satisfies a relaxation of the triangle inequality (defined below), computing the max element in the first column and sampling $n^{0.1}$ entries of each column uniformly at random suffices to obtain $n^{0.9}$-approximate estimates to the the column norms. A similar guarantee holds for approximating the row norms.
Using these coarse estimators as a proxy for the column norms requires us 
to oversample columns by a $n^{0.9}$-factor. Therefore, resulting matrix, $\AA_{(1)}$ 
is not small enough to compute a singular value decomposition or even use 
input-sparsity time methods and obtain a sublinear running time. We then try to compute row norms of the sub-sampled matrix $\AA_{(1)}$ and hope that the resulting matrix, $\AA_{(2)}$ is small enough to run known low-rank approximation algorithms and obtain an overall sublinear time algorithm.

However, the above approach presents significant challenges. First, after 
sampling and rescaling the columns of the input matrix, the resulting 
smaller matrix $\AA_{(1)}$ need to be a distance matrix or satisfy triangle inequality. Next, it is unclear how to relate the best rank-$k$ 
approximation for the smaller sub-sampled matrix $\AA_{(2)}$ to the best 
rank $k$ approximation for $\AA$ as 
they do not even have the same dimension. To address the first issue, we 
form geometrically increasing buckets for the scaling factors of each row 
and restricted to each bucket, $\AA_{(1)}$ satisfies a relaxation of the 
triangle inequality. Using this decomposition, we show we can estimate 
row norms of $\AA_{(1)}$ in sublinear time. 

The second challenge requires us to show a structural result about sub-
sampling rows and column via coarse estimates to column norms. Here, we 
introduce \emph{additive-error projection-cost preserving sampling} and 
show that sampling using our coarse estimates to column norms suffice to 
construct a matrix $\AA_{(1)}$ such that it preserve all rank-$k$ 
subspaces in the column space of $\AA$ i.e. for any rank-$k$ projection matrix $\X$,  
\begin{equation*}
\| \AA - \X\AA \|^2_F - \epsilon \|\AA\|^2_F \leq \| \AA_{(1)} - \X \AA_{(1)}   \|^2_F \leq \| \AA - \X\AA \|^2_F + \epsilon \|\AA\|^2_F
\end{equation*}
A similar guarantee holds for the row space of $\AA$. We note that the sampling guarantees we achieve are a relaxation of the \emph{projection-cost preserving sampling} guarantee first introduced in \cite{cohen2015dimensionality} and further 
improved in \cite{cmm17}. Formally, for any rank-$k$ projection matrix $\X$ 
\begin{equation*}
(1-\epsilon)\| \AA - \X\AA \|^2_F \leq \| \AA_{(1)} - \X \AA_{(1)}   \|^2_F \leq (1+\epsilon) \| \AA - \X\AA \|^2_F 
\end{equation*}
However, these guarantees are relative-error and their approach relies on approximating the ridge leverage scores of $\AA$ which require input-sparsity time. We note that \emph{additive-error projection-cost preserving sampling} could be of independent interest in future work.

An important consequence of \emph{additive-error projection-cost preserving sampling} is that an approximately optimal projection for $\AA_{(1)}$ is an approximately optimal projection for $\AA$. This is a key link stating that solving the problem for a smaller matrix helps to recover a solution for a larger one. In particular, let $\X_{(1)}$ be a matrix such that $\|\AA_{(1)} - \X_{(1)} \AA_{(1)} \|^2_F \leq \min_{\X} \|\AA_{(1)} - \X \AA_{(1)} \|^2_F  + \epsilon \|\AA_{(1)} \|_F$, then 
\begin{equation*}
\|\AA- \X_{(1)} \AA \|^2_F \leq \min_{\X} \|\AA - \X \AA \|^2_F + \epsilon \|\AA \|^2_F
\end{equation*}
A similar guarantee holds for the row space. Armed with how to relate the optimal solution for a smaller matrix with the optimal solution for a larger one, we decrease the row dimension of $\AA_{(1)}$ by constructing an \emph{additive-error projection-cost preserving sketch}, $\AA_{(2)}$, that preserves all the rank-$k$ subspaces in the row space of $\AA_{(1)}$ i.e. for all rank-$k$ projections $\X$
\begin{equation*}
\|\AA_{(1)} - \AA_{(1)}\X\|^2_F - \epsilon \|\AA_{(1)} \|^2_F \leq \|\AA_{(2)} - \AA_{(2)}\X  \|^2_F \leq  \|\AA_{(1)} - \AA_{(1)}\X\|^2_F + \epsilon \|\AA_{(1)} \|^2_F
\end{equation*}
Again, we can show that an approximately optimal projection for the rows of $\AA_{(2)}$ is an approximately optimal projection for the rows of $\AA_{(1)}$. In particular let $\X_{(2)}$ be a matrix 
such that $\|\AA_{(2)} -  \AA_{(2)}\X_{(2)} \|^2_F \leq \min_{\X} \|\AA_{(2)} -\AA_{(2)} \X  \|^2_F  + \epsilon \|\AA \|_F$, then 
\begin{equation*}
\|\AA_{(1)} -  \AA_{(1)}\X_{(2)} \|^2_F \leq \min_{\X} \|\AA_{(1)} - \AA_{(1)}\X \|^2_F + \epsilon \|\AA_{(1)} \|^2_F
\end{equation*}

Indeed, we show that an approximately optimal rank-$k$ projection, $\X_{(2)}$, for $\AA_{(2)}$ can be computed in $\widetilde{O}( (m^{1.34} + n^{1.34}) \textrm{poly}(k \epsilon^{-1}) )$ time (sublinear in the size of the input) using input-sparsity low-rank approximation, introduced in \cite{clarkson2013low}. Further, $\X_{(2)}$ is an approximately optimal projection for the row space of $\AA_{(1)}$. However, observe that $\X_{(2)}$ projects the rows of $\AA_{(1)}$ onto a $k$-dimensional subspace, while $\AA_{(1)}$ preserves all column projections of $\AA$. We need to find a matrix that projects the columns of $\AA_{(1)}$ on to a rank-$k$ subspace and is approximately optimal. 

To this end, we begin with computing a row space, $\V_{(2)}$ for $\AA_{(2)}\X_{(2)}$. Since $\AA_{(1)}\X_{(2)}$ is an approximately optimal solution for $\AA_{(1)}$, we observe the row space $\V_{(2)}$ contains such a solution. Therefore we set up the following regression problem 
\begin{equation*}
\min_{\X} \| \AA_{(1)} - \X\V_{(2)} \|_F
\end{equation*}
We show that a solution to the above regression problem enables us to compute a column space $\U_{(1)}$ such that it contains an approximately optimal solution for $\AA$. Note, it was essential to find a column space in order to use the column projection preserving property that relates $\AA_{(1)}$ and $\AA$. Given such a column space we can set up another regression problem 
\begin{equation*}
\min_{\X} \| \AA - \U_{(1)}\X \|_F
\end{equation*}
and solving this provides a low-rank approximation for $\AA$. However, we observe that the dimensions involved in the regression problems are too large and the running time is no longer sublinear. To decrease the running time, we sketch the regression problem and solve it approximately. This approach shows up as a subroutine in \cite{mw17} and was previously studied in \cite{clarkson2013low} and  \cite{drineas2008relative}. 
In particular, we use leverage score sampling to construct a sketching matrix $\E$ such that $\min_{\X} \| \AA_{(1)}\E - \X\V_{(2)}\E \|^2_F$ can be solved in sub-linear time and the solution is a $(1+\epsilon)$-error relative approximation to the original regression problem. Therefore, all the the steps involved in our algorithm are sublinear and we obtain an algorithm for low-rank approximation that runs in $\widetilde{O}( (m^{1.34} + n^{1.34}) \textrm{poly}(k \epsilon^{-1}))$.  

Lastly, we focus on improving the exponent in the above run time. The critical observation 
here is that we can recursively sketch rows and columns of $\AA$ such that all the rank-$k$ 
projections are preserved. To this end, we show that a recursive \emph{ projection-cost preserving sampling} guarantee holds, i.e. we can indeed alternate between preserving row and column such that all rank-$k$ projections across all levels can be preserved simultaneously.
At a high level, the algorithm is then to recursively sub-sample 
columns and rows of $\AA$ such that we obtain \emph{projection-cost preserving sketches} at 
each step. However, the resulting sub-matrices need not be a distance 
matrix and the naive bucketing from the previous algorithm does not work. 
Recall, we were able to partition the matrix into $O(\log(n) 
\epsilon^{-1})$ sub-matrices such that each one satisfied approximate 
triangle inequality. We show that we can recursively apply the bucketing 
to each sub-matrix independently and maintain a handle on the blow-up in 
the approximation guarantee.

We note that we can recurse only a constant number of times and show that this is sufficient to obtain a matrix that has dimensions 
independent of $m$ and $n$. We can then quickly compute the best rank-$k$ approximation to such a 
matrix using SVD. We now start with the SVD solution and follow the previously described 
approach to recurse all the way up. At each level, we solve a regression problem to switch 
between finding approximately optimal row and column spaces, up to additive error, and controlling the propagation of the error as we recurse all the way to the top.
Intuitively, it is important to 
preserve all rank-$k$ subspaces since we do not know which projection will be approximately 
optimal when we recurse back up. We show that the above algorithm runs in $\widetilde{O}( 
(m^{1+\gamma} + n^{1+\gamma})\textrm{poly}(k \epsilon^{-1}))$, for a small constant $\gamma$.

\begin{comment}
We now describe a few key challenges in the above result. First, observe 
that at each level of the recursing, the algorithm sub-samples and 
rescales rows and columns. The resulting matrix need not be a distance 
matrix and the naive bucketing from the previous algorithm does not work. 
Recall, we were able to partition the matrix into $O(\log(n) 
\epsilon^{-1})$ sub-matrices such that each one satisfied approximate 
triangle inequality. We show that we can recursively apply the bucketing 
to each sub-matrix independently and maintain a handle on the blow-up in 
the approximation guarantee. However, we can only recurse a constant 
number of times and we show this is sufficient. Next, we show that a recursive \emph{ projection-cost preserving sampling} guarantee holds, i.e. we can indeed alternate between preserving row and column such that all rank-$k$ projections across all levels can be preserved simultaneously. Finally, we show that if we start with a an optimal rank-$k$ approximation at the bottom of the recursion, we can solve regression problems at each level to alternate between row and column spaces, s  
\end{comment}

%% file: first_sublinear.tex
\iffalse
\begin{theorem}
\label{thm:sublinear_adm_lra}
Given a Distance Matrix $\AA \in \mathbb{R}^{n \times n}$, such that the underlying points lie in an arbitrary metric space, any $\epsilon > 0$, a small constant $\gamma > 0$, Algorithm $1$ accesses $O\left(n^{1+\gamma}\right)$ entries of $\AA$ and runs in time $O\left(n^{1+\gamma}\text{poly}(\frac{k}{\epsilon}) )$ to output matrices $\MM, \NN \in \mathbb{R}^{n \times k}$ such that with probability at least $9/10$,
\[
\left\| \AA - \MM \NN^T \right\|^2_{F} \leq  \left\| \AA - \AA_k\right\|^2_{F} + \epsilon \left\| \AA \right\|^2_F 
\]
where $\AA_k = \argmin_{\text{rank-k }\BB} \left\| \AA - \BB\right\|_F$.
\end{theorem}
\fi

\section{Row and Column Norm Estimation}

\begin{Frame}[\textbf{Algorithm \ref{alg:row_norm_estimation} : Row Norm Estimation.}]
\label{alg:row_norm_estimation}
\textbf{Input:} A Distance Matrix $\AA_{m \times n}$, Sampling parameter $b$.
\begin{enumerate}
	\item Let $x = \argmin_{i \in [m]} \AA_{i,1}$.
	\item Let $d = \max_{j \in [n]} \AA_{x, j}$.
    \item For $i \in [m]$, let $\mathcal{T}_i$ be a uniformly random sample of $\Theta(b)$ indices in $[n]$.
    \item $\widetilde{X}_i = d^2 + \sum_{j \in \mathcal{T}_i} \frac{n}{b}\AA^2_{i,j}$. 
\end{enumerate}
\textbf{Output:} Set $\{\widetilde{X}_1, \widetilde{X}_2, \ldots \widetilde{X}_m\}$
\end{Frame}

We observe that we can obtain a rough estimate for the row or column norms of a distance matrix by uniformly sampling a small number of elements of each row or column. The only structural property we need to obtain such an estimate is approximate triangle inequality. 

\begin{definition}(Approximate Triangle Inequality.)
Let $\AA$ be a $m \times n$ matrix. Then, matrix $\AA$ satisfies approximate triangle inequality if, for any $\epsilon \in [0, 1]$, for any $p \in [m]$, $q,r \in [n]$   
\begin{equation}
\label{eqn:approx_tri}
\frac{|\AA_{p,r} - \max_{i \in [m]}|\AA_{i,q} - \AA_{i,r}||}{(1+\epsilon)} \leq \AA_{p,q} \leq (1+\epsilon)\left(\AA_{p,r} + \max_{i \in [m]}|\AA_{i,q} - \AA_{i,r}| \right)
\end{equation}
and
\begin{equation}
\label{eqn:approx_tri_2}
\frac{ |\AA_{p,q} - \AA_{p,r}|}{1+\epsilon}\leq \max_{i \in [m]}|\AA_{i,q} - \AA_{i,r}| \leq (1+\epsilon)\left( \AA_{p,q} + \AA_{p,r} \right)
\end{equation}
Further, similar equations hold for $\AA^T$.
\end{definition}

The above definition captures distance matrices if we set $\epsilon=0$. In order to see this, 
recall, each entry in a $m \times n$ matrix $\AA$ is associated with a distance between 
points sets $\mathcal{P}$ and $\mathcal{Q}$, such that $|\mathcal{P}|=m$ and 
$|\mathcal{Q}|=n$. Then, for points $p \in \mathcal{P}$, $q \in \mathcal{Q}$, $\AA_{p,q}$ 
represents $d(p,q)$, where $d$ is an arbitrary distance metric. Further, for arbitrary point, 
$i \in \mathcal{P}$ and $r \in \mathcal{Q}$, $\max_{i}|\AA_{i,q} - \AA_{i,r}| = \max_{i} 
|d(i,r) - d(i,q)|$. Intuitively, we would like to highlight that, for the case where $\AA$ is a distance matrix, $\max_{i \in [m]}|\AA_{i,q} - \AA_{i,r}|$ represents a lower bound on the distance $d(q,r)$. 
Since $d$ is a metric, it follows triangle inequality, and 
$d(p,q) \leq d(p,r) + d(q,r)$. Further, by reverse triangle inequality, for all $i \in [m]$, $d(q,r) \geq |d(i,q) - d(i, r)|$. Therefore, $\AA_{p,q} \leq \AA_{p,r} + \max_{i}|\AA_{i,q} - \AA_{i,r}|$ and distance matrices satisfy equation \ref{eqn:approx_tri}. Next, $\max_{i \in [m]}|\AA_{i,q} - \AA_{i,r}| = \max_{i \in [m]}|d_{i,q} - d_{i,r}| \leq d(q,r)$,  and 
$d(q, r) \leq d(p,r) + d(p,q) = \AA_{p,r} + \AA_{p,q}$ therefore, equation \ref{eqn:approx_tri_2} is satisfied. We note that approximate triangle inequality is a relaxation of the traditional triangle 
inequality and is sufficient to obtain coarse estimates to row and column norms of $\AA$ in 
sublinear time.

\begin{lemma}(Row Norm Estimation.)
\label{lem:row_norm_estimation}
Let $\AA$ be a $m \times n$ matrix such that $\AA$ satisfies approximate triangle inequality. 
For $i \in [m]$ let $\AA_{i,*}$ be the $i^{th}$ row of $\AA$. Algorithm 
\ref{alg:row_norm_estimation} uniformly samples $\Theta(b)$ elements from $\AA_{i,*}$ and 
with probability at least $9/10$ outputs an estimator which obtains an $O\left(\frac{n}
{b}\right)$-approximation to $\left\| \AA_{i,*} \right\|^2_2$. Further, Algorithm 
\ref{alg:row_norm_estimation} runs in $O(bm + n)$ time.
\end{lemma}
\begin{proof}
It is easy to analyze the running time of Algorithm \ref{alg:row_norm_estimation}. Step $1$ 
and $2$ run in $O(n + m)$ as they correspond to reading a column and a row. Uniformly 
sampling $\Theta(b)$ indices for each row takes $O(bm)$ time. Overall, we get a running time 
of $O\left(bm + n\right)$.

By reading the first column of $\AA$, we obtain a entry $\AA_{x,1}$ such that it is the 
minimum entry of the first column, i.e. $x = \argmin_{i\in[m]} \AA_{i,1}$. Then, reading row 
$\AA_{x,*}$, we obtain a entry $\AA_{x,y}$ such that index $y = \argmax_{j \in [n] } \AA_{x, 
j}$ i.e., $\AA_{x,y}$ is the largest entry in row $x$. Let $d$ denote the entry $\AA_{x,y}$ 
of the input matrix. Further, let $d_{\textrm{max}} =  \max_{j,j' \in [n]} \max_{i\in [m]}| \AA_{i,j} - \AA_{i,j'}|$. Note, we cannot compute $d_{\textrm{max}}$ without reading all the 
entries of the matrix and this is no longer sublinear time. We should think of $d_{\max}$ as representing the diameter of $\mathcal{Q}$ when $\AA$ is a distance matrix. To see why this is true, observe, for $q_j, q_{j'} \in \mathcal{Q}$, $\max_{i\in [m]}| \AA_{i,j} - \AA_{i,j'}|$ represents the best lower bound on $d(q_j,q_{j'})$, via reverse triangle inequality. Finding the largest such lower bound over all pairs of points in $\mathcal{Q}$ is the best lower bound on the diameter of $\mathcal{Q}$.  

Intuitively, we show that $d$ is 
a good proxy for $d_{\textrm{max}}$. Recall, for row $\AA_{i,*}$, Algorithm 
\ref{alg:row_norm_estimation} outputs the estimator $\widetilde{X}_i = d^2 + \frac{n}{b} 
\sum_{\ell \in \mathcal{T}_i} \AA^2_{i,\ell}$, where $\mathcal{T}_i$ is a uniform sample of 
indices in the range $[1,n]$, such that $|\mathcal{T}_i| =\Theta(b)$. 
Let us first consider the case of estimating the norm of $\AA_{x,*}$. Note, Algorithm 
\ref{alg:row_norm_estimation} reads the entire row and can compute the norm exactly. However, 
the analysis of our estimator is more intuitive in this case. Observe, by approximate 
triangle inequality and the definition of $d_{\max}$,

\begin{equation}\AA_{x,j} \leq (1+\epsilon)\left( \AA_{x,y} + \max_{i \in [m]}|\AA_{i,y} - \AA_{i,j}| \right) \leq (1+\epsilon)(d + 
d_{\max})\end{equation}
where first inequality follows from the upper bound in \ref{eqn:approx_tri}, the second 
follows from recalling that $\AA_{x,y} = d$ and observing that $\max_{i \in [m]}|\AA_{i,y} - \AA_{i,j}|$ is upper bounded 
by $d_{\max}$. Similarly, using the lower bound in \ref{eqn:approx_tri},
\begin{equation} \AA_{x,j} \geq \frac{| \AA_{x,y} - \max_{i \in [m]}|\AA_{i,y} - \AA_{i,j}| |}{1+\epsilon} \geq \frac{|d - 
d_{\max}|}{(1+\epsilon)}
\end{equation}

Intuitively, if $d_{\max}$ is sufficiently larger than $d$, $\AA_{x,j}$ is within a constant 
factor of $d_{\max}$. Formally, if $d_{\max} \geq 2 d$, for all $j \in [n]$, $\frac{d_{\max}}
{2(1+\epsilon)} \leq \AA_{x,j} \leq \frac{3(1+\epsilon)d_{\max}}{2}$. Therefore, each entry 
in $\AA_{x,*}$ is at least $d_{\max}/2(1+\epsilon)$. Observe, by linearity of expectation, 
$\expec{}{\widetilde{X}_x} = d^2 + \left\| \AA_{x,*} \right\|^2_2$. Therefore, by Markov's 
bound we have that with probability at least $1 -1/c$, \begin{equation}\widetilde{X}_x \leq 
c(d^2 + \left\| \AA_{x,*} \right\|^2_2) \leq 5c \left\| \AA_{x,*} \right\|^2_2\end{equation} 
where the last inequality follows from $\left\| \AA_{x,*} \right\|^2_2 \geq \frac{nd^2}
{4(1+\epsilon)^2}$ (since each entry of $\AA_{x,*}$ is at least $d/2(1+\epsilon)$). 
Further, 
\begin{equation}\widetilde{X}_x \geq d^2 + \left(\frac{n}{b}\right)\left(\frac{bd^2}{4(1+\epsilon)^2}\right) \geq \frac{\left\| \AA_{x,*} \right\|^2_2}{160}\end{equation}
where the first inequality follows from each sampled entry being at least $\frac{d}{2(1+\epsilon)}$ and the second inequality follows from $\left\| \AA_{x,*} \right\|^2_2 \leq \left(\frac{3(1+\epsilon)d}{2}\right)^2 n \leq  10nd^2$. For an appropriate value of $c$, $\widetilde{X}_x = \Theta\left(\left\| \AA_{x,*} \right\|^2_2\right)$ with probability at least $99/100$.

Note, $d$ could be really small compared to $d_{\max}$ and we cannot hope for a constant approximation to the row norm. Formally, $d \leq d_{\max}$ and we consider two cases, one where $\left\| \AA_{x,*} \right\|^2_2 \leq \frac{nd^2}{b}$ and the second being its complement. In the first case, we observe, 
\begin{equation}\widetilde{X}_x \geq d^2 \geq \frac{b}{n}\left\| \AA_{x,*} \right\|^2_2\end{equation}
where the first inequality follows from $\left\| \AA_{x,*} \right\|^2_2 \geq 0$ and the second follows from $\left\| \AA_{x,*} \right\|^2_2 \leq \frac{nd^2}{b}$. On the other hand, observe, by linearity of expectation, $\expec{}{\widetilde{X}_x} = d^2 + \left\| \AA_{x,*} \right\|^2_2$. By Markov's bound, with probability at least $1 - \frac{1}{c}$, 
\begin{equation}\widetilde{X}_x \leq c (d^2 + \left\| \AA_{x,*} \right\|^2_2) \leq 2c\left\| \AA_{x,*} \right\|^2_2\end{equation}
where the last inequality follows from $d$ being an element of $\AA_{x,*}$.
Therefore, combining the upper and lower bound, with probability at least $99/100$, $\widetilde{X}_x$ achieves a $\Theta\left(\frac{n}{b}\right)$-approximation to $\left\| \AA_{x,*} \right\|^2_2$.  

Next, consider the second case where $\left\| \AA_{x,*} \right\|^2_2 \geq \frac{nd^2}{b}$. We begin by bounding $\var{\widetilde{X}_x}$. Let $Y_{x,j} = \AA^2_{x,j}$ with probability $b/n$, and $Y_{x,j} = 0$ otherwise. Then,
\begin{equation}
\label{eqn:var}
\begin{split}
\var{\widetilde{X}_x} = \varf{d^2 + \frac{n}{b} \sum_{j \in \mathcal{T}_x} Y_{x,j} } & =   \left(\frac{n}{b}\right)^2\varf{\sum_{j \in \mathcal{T}_x} Y_{x,j}} \\
&\leq \left(\frac{n}{b}\right)^2 \expecf{}{ (\sum_{j \in \mathcal{T}_x} Y_{x,j})^2} \\ 
& \leq \left(\frac{n}{b}\right)^2 \frac{b}{n}\left( \frac{\left\| \AA_{x,*} \right\|^2_2}{d^2}\right) d^4  \\
& \leq \frac{\left\| \AA_{x,*} \right\|^4_2}{100c} 
\end{split}
\end{equation}
where $c$ is another fixed constant, the first inequality follows from the definition of variance, the second follows from $Y_{x,j}$ being upper bounded by $d^2$ and an averaging argument, and the last follows from the assumption that $\frac{nd^2}{b} \leq \|\AA_{x,*} \|^2_2$. Observe that the variance is maximized when there are $\frac{\left\| \AA_{x,*} \right\|^2_2}{d^2}$ entries with value $d$ and the last inequality follows from our assumption in this case. By Chebyshev's inequality, 
\begin{equation}
    \prob{}{\left|\widetilde{X}_x - \expecf{}{\widetilde{X}_x}\right| > \frac{\left\| \AA_{x,*} \right\|^2_2}{100}} \leq \frac{\varf{\widetilde{X}_x}}{c^2\left\| \AA_{x,*} \right\|^4_2} \leq \frac{1}{c^2}
\end{equation}
Therefore, with probability at least $1 - 1/c^2$, \begin{equation}\widetilde{X}_x  = \expec{}{\widetilde{X}_x} \pm \frac{\left\| \AA_{x,*} \right\|^2_2}{100} = d^2 + \left(1 \pm \frac{1}{100}\right)\left\| \AA_{x,*} \right\|^2_2\end{equation}
Therefore, $\widetilde{X}_x$ achieves a $\Theta(1)$-approximation to $\left\| \AA_{x,*} \right\|^2_2$. It follows that in every case, $\widetilde{X}_x$ achieves an $O\left(\frac{n}{b}\right)$-approximation to $\left\| \AA_{x,*} \right\|^2_2$ with probability at least $9/10$. 

Now we analyze our estimator in its full generality by considering row $\left\| \AA_{i,*} \right\|^2_2$ for any $i \neq x$. We define $d' = \max_{j \in [n]} \AA_{i, j}$ to be the largest entry in row $i$. Note, we do not explicitly know $d'$ as this would require reading the entire row. Instead, we show that biasing our estimator with $d^2$ suffices for all rows. Recall, $\AA_{x,y} =  \max_{j \in [n]}\AA_{x, j}$. 
We follow an analysis similar to the simplified one above. Since our estimator is still biased by $d^2$, we analyze the cases where $d$ is small or large compared to $d_{\max}$ separately. Consider the first case, where $d \geq 8 d_{\max}$. We begin by bounding $d'$ in terms of $d$ and $d_{\max}$. 
By the approximate triangle inequality, 
\begin{equation}
\AA_{x, 1} \geq  \frac{| \AA_{x, y} - \max_{i\in[m]}|\AA_{i,1} - \AA_{i,y}| |}{1+\epsilon} \geq \frac{|d - d_{\max}|}{1+\epsilon} \geq \frac{7}{1+\epsilon}d_{\max} \geq 3d_{\max}
\end{equation}
where the second inequality follows from recalling the definition of $d$ and observing that $d_{\max} \geq$  $\max_{i\in[m]}|\AA_{i,1} - \AA_{i,y}|$, and the third inequality follows from the assumption in this case. Alternatively, we can repeat the above bound to get 
\begin{equation}
\AA_{x, 1} \geq  \frac{| \AA_{x, y} - \max_{i\in[m]}|\AA_{i,1} - \AA_{i,y}| |}{1+\epsilon} \geq \frac{|d - d_{\max}|}{1+\epsilon} \geq \frac{7d}{(1+\epsilon)8} \geq \frac{7d}{16}
\end{equation}
Further, 
\begin{equation}
\AA_{x, 1} \leq \AA_{i,1} \leq d'
\end{equation}
where the first inequality follows from the definition of $\AA_{x,1}$ and the second inequality follows from the definition of $d'$. Therefore, combining the two equations above, we get that $d' \geq 3 d_{\max}$ and $d' \geq \frac{7d}{16}$. Next, we show a bound of any entry of the $i$-th row. By the upper bound in approximate triangle inequality, 
\begin{equation}
\AA_{i,j} \leq (1+\epsilon) \left( \AA_{i,j'} + \max_{i'} |\AA_{i',j} - \AA_{i',j'}| \right)
\leq (1+\epsilon) \left( d' + d_{\max} \right) \leq \frac{8d'}{3}
\end{equation}
where the second inequality follows from $d'$ being the largest element in row $i$, and the definition of $d_{\max}$, and the last follows from $d' \geq 3 d_{\max}$. Similarly, by the lower bound in approximate triangle inequality,
\begin{equation}
\AA_{i,j} \geq \frac{ | \AA_{i,j'} - \max_{i'} |\AA_{i',j} - \AA_{i',j'}| |}{1+\epsilon}
\geq \frac{| d' - d_{\max}|}{1+\epsilon} \geq \frac{d'}{3}
\end{equation}
Combining the two equations above, for all $j \in [n]$, $\AA_{i, j} = \Theta(d')$,  all entries of $ \AA_{i,*} $ are within a constant factor of each other. Therefore, $\| \AA_{i, *}\|^2_2 \leq \frac{64 n d'^2}{9}$. 
%Then,
%\begin{equation}
%\| \AA_{i, *}\|^2_2 \leq \frac{64 n d'^2}{9}
%\end{equation}
Recall, our estimator is 
\begin{equation}
\widetilde{X}_i = d^2 + \sum_{j \in \mathcal{T}_i} \frac{n}{b}\AA^2_{i,j}\geq \frac{n}{b}\frac{bd'^2}{9} \geq \frac{\left\| \AA_{i,*}\right\|^2_2}{64}
\end{equation}
We observe that by linearity of expectation, $\expec{}{\widetilde{X}_i} = d^2 + \left\| \AA_{i,*}\right\|^2_2 $. By Markov's we know that with probability at least $1 - 1/c$, 
\begin{equation}
\widetilde{X}_i \leq c \left( d^2 + \left\| \AA_{i,*}\right\|^2_2\right) \leq c \left( \frac{216}{49}d'^2 + \left\| \AA_{i,*}\right\|^2_2\right) \leq c'\left\| \AA_{i,*}\right\|^2_2 
\end{equation}
where the second inequality follows from recalling that $d' \geq \frac{7d}{16}$ and the last inequality follows from $d'$ being an entry in $\AA_{i,*}$. Therefore, $\widetilde{X}_i = \Theta\left(\left\| \AA_{i,*} \right\|^2_2\right)$, with probability at least $99/100$.

Now we analyze the case where $d \leq 8 d_{\max}$. We then consider two cases: $\left\| \AA_{i,*} \right\|^2_2 \geq \frac{nd'^2}{b}$ or $\left\| \AA_{i,*} \right\|^2_2 \leq \frac{nd'^2}{b}$. In the first case, computing the variance exactly as in equation \ref{eqn:var} and applying Chebyshev's inequality, we get that with probability at least $1 -1/c^2$, 
\begin{equation}\widetilde{X}_i = d^2 + \left(1 \pm \frac{1}{100}\right)\left\| \AA_{i,*} \right\|^2_2\end{equation} 
Since $d \leq 8 d_{\max} $ and $\frac{1}{16}d^2_{\textrm{max}} \leq \left\| \AA_{i,*} \right\|^2_2$, $\widetilde{X}_i = \Theta\left(\left\| \AA_{i,*} \right\|^2_2\right)$ with probability $99/100$. 

In the second case, we show that all entries of $\AA_{i,*}$ are within a constant factor of each other.  Recall, $\widetilde{X}_i \geq d^2$ since $\left\| \AA_{i,*} \right\|^2_2 \geq 0$. If $8d \geq d'$,  \begin{equation}\widetilde{X}_i \geq \frac{d'^2}{64} \geq \frac{b \left\| \AA_{i,*} \right\|^2_2}{64n}\end{equation} 
Recall, $\expec{}{\widetilde{X}_i} = d^2 + \| \AA_{i,*} \|^2_2 $, $d \leq \frac{16}{7} d'$ and $d'^2 \leq  \| \AA_{i,*} \|^2_2$. Therefore, the upper bound on $\widetilde{X}_i$ follows from Markov's bound and holds with probability at least $99/100$. 
If instead $8d \leq d'$, we show all entries of $\AA_{i,*}$ are within a constant factor of each other. To see this, let $d' = \AA_{i, j^*}$ and observe by approximate triangle inequality, 
\begin{equation}
\begin{split}
\AA_{i, j} \leq (1+\epsilon)\left(\AA_{i, j^*}  + \max_{i'}|\AA_{i', j} - \AA_{i', j^*} | \right) & \leq 2d' + 2(1+\epsilon)(\AA_{x, j^*} + \AA_{x,j} ) \\ 
& \leq 2d' + 8d \leq 3d'
\end{split}
\end{equation}
where the first and second inequalities follow from initially applying the upper bound in equation \ref{eqn:approx_tri} to $\AA_{i,j}$, and then applying the equation \ref{eqn:approx_tri_2} to $\max_{i'}|\AA_{i', j} - \AA_{i', j^*} |$, the third inequality follows from upper bounding both $\AA_{x, i^*}$ and $\AA_{x,j}$ by $d$, and the last inequality follows from the assumption $8d \leq d'$ combined with the definition of $d'$. Further,
\begin{equation}
\begin{split}
\AA_{i,j} \geq \frac{\AA_{i, i^*} - \AA_{i^*j}}{1+\epsilon} & \geq \frac{1}{2} (d'  - 2( \AA_{ x, j^*} + \AA_{x, j} )) \\
& \geq \frac{1}{2} d'  - 2d \geq \frac{d'}{4}
\end{split}
\end{equation}
where the third inequality follows from upper bounding both $\AA_{x, i^*}$ and $\AA_{x,j}$ by $d$, and the last inequality follows from the assumption $8d \leq d'$ combined with the definition of $d'$.
Therefore, $d(p_i, q_j) = \Theta(d')$. Finally, $\expec{}{\widetilde{X}_i} = d^2 + \| \AA_{i,*} \|^2_2 $ and by Markov's bound, our estimator 
$\widetilde{X}_i = \Theta\left(d^2 + \left\| \AA_{i,*} \right\|^2_2\right)$ with probability $99/100$. Since $d^2 \leq \left\| \AA_{i,*} \right\|^2_2$, we obtain a constant factor approximation with probability at least $99/100$. By a union bound over all the probabilistic events, all of them simultaneously hold with probability at least $9/10$, which finishes the proof.  
\end{proof}

To obtain an $O\left(\frac{n}{b}\right)$ approximation for all the $m$ rows simultaneously with constant probability, we can compute $O\left(\log(m) \right)$ estimators for each row and take their median. We also observe that Column and Row Norm Estimation are symmetric operations and a slight modification to Algorithm $\ref{alg:row_norm_estimation}$ yields a Column Norm Estimation algorithm with the following guarantee:

\begin{corollary}(Column Norm Estimation.)
\label{cor:column_norm_estimation}
Let $\AA$ be a $m \times n$ matrix such that $\AA$ satisfies approximate triangle inequality. For $j \in [n]$ let $\AA_{*,j}$ be the $j^{th}$ column of $\AA$. There exists an algorithm that uniformly samples $\Theta(b)$ elements from $\AA_{*,j}$ and with probability $9/10$ outputs an estimator which is an $O\left(\frac{m}{b}\right)$-approximation to $\left\| \AA_{*,j} \right\|^2_2$. Further, this algorithm runs in $O(bn+m)$ time. 
\end{corollary} 

\input{pcp.tex}

\section{A Sublinear Time Algorithm.}

\begin{Frame}[\textbf{Algorithm \ref{alg:first_sublinear} : First Sublinear Time Algorithm.}]
\label{alg:first_sublinear}
\textbf{Input:} A Distance Matrix $\AA_{m \times n}$, integer $k$ and $\epsilon >0$. 
\begin{enumerate}
	\item  Set $b_1 =\frac{\epsilon n^{0.34}}{\log(n)} $ and $b_2 = \frac{\epsilon m^{0.34}}{\log(m)} $. 
  Set $s_1 =\Theta\left( \frac{mk^2\log(m)}{b_1\epsilon^2} \right)$ and $s_2 = \Theta\left( \frac{nk^2\log(n)}{b_2\epsilon^2} \right)$.
	\item Let $\widetilde{X}_j$ be the estimate for $\|\AA_{*,j}\|^2_2$ returned by \texttt{ColumnNormEstimation}$(\AA, b_1)$. Recall, $\widetilde{X}_j$ is an $O\left(\frac{m}{b_1}\right)$-approximation to $\AA_{*,j}$. 
     \item Let $q = \{ q_1, q_2 \ldots q_n \} $ denote a distribution over columns of $\AA$ such that $q_i = \frac{\widetilde{X}_{j}}{\sum_j \widetilde{X}_{j}} \geq \left(\frac{b_1}{m} \right)\frac{\| \AA_{*,j} \|^2_2}{\| \AA \|^2_F}$. Construct a \textit{column pcp} for $\AA$ by sampling $s_1$ columns of $\AA$ such that each column is set to $\frac{\AA_{*,j}}{\sqrt{s_1 q_j}}$ with probability $q_j$. Let $\AA\S$ be the resulting $m \times s_1$ matrix that follows guarantees of Theorem \ref{thm:pcp_main_col}.
     \item To account for the rescaling, consider $O(\epsilon^{-1}\log(n))$ weight classes for scaling parameters of the columns of $\AA\S$. Let $\AA\S_{|\mathcal{W}_g}$ be the columns of $\AA\S$ restricted to the weight class $\mathcal{W}_g$ (defined below.)
     \item Run the $\texttt{RowNormEstimation}(\AA\S_{|\mathcal{W}_g}, b_2)$ estimation algorithm with parameter $b_2$ for each weight class independently and sum up the estimates for a given row. Let $\widetilde{X}_i$ be the resulting $O\left(\frac{n}{b_2}\right)$-approximate estimator for $\AA\S_{i,*}$. 
     \item Let $p = \{ p_1, p_2, \ldots p_m \}$ denote a distribution over rows of $\AA\S$ such that $p_i =  \frac{\widetilde{X}_{i}}{\sum_i \widetilde{X}_{i}} \geq \left( \frac{b_2}{n}\right)\frac{\| \AA\S_{i,*} \|^2_2}{\| \AA\S \|^2_F}$. Construct a \textit{row pcp} for $\AA\S$ by sampling $s_2$ rows of $\AA\S$ such that each row is set to $\frac{\AA\S_{i,*}}{\sqrt{s_2 p_i}}$ with probability $p_i$. Let $\T\AA\S$ be the resulting $s_2 \times s_1$ matrix that follows guarantees of Corollary \ref{thm:pcp_main_row}. 
     \item Run the input-sparsity time low-rank approximation algorithm (corresponding to Theorem \ref{thm:clarkson_woodruff}) on $\T\AA\mathbf{S}$ with rank parameter $k$ to obtain a rank-$k$ approximation to $\T\AA\S$, output in factored form: $\LL , \mathbf{D}, \mathbf{W}^T$. Note, $\LL \mathbf{D}$ is an $s_2 \times k$ matrix and $\mathbf{W}^T$ is a $k \times s_1$ matrix.
     \item Consider the regression problem $\min_{\X} \| \AA\S -\X\WW^T \|^2_F$. Sketch the problem using the leverage scores of $\WW^T$ as shown in Theorem \ref{thm:clarkson_woodruff_regression} to obtain a sampling matrix $\E$ with poly$(\frac{k}{\epsilon})$ columns. Compute $\X_{\AA\S} = \argmin_{\X}\| \AA\S\E - \X\WW^T\E \|^2_F$. Let $\X_{\AA\S}\WW^T = \P'\NN'^T$ be such that $\P'$ has orthonormal columns.
     \item Consider the regression problem $\min_{\X} \|\AA -\P'\X\|^2_F$. Sketch the problem using the  the leverage scores of $\P'$ following Theorem \ref{thm:clarkson_woodruff_regression} to obtain a sampling matrix $\E'$ with poly$(\frac{k}{\epsilon})$ rows. Compute $\X_{\AA} = \argmin_{\X} \|\E'\AA - \E'\P'\X \|^2_F$. 
\end{enumerate}
\textbf{Output:} $\MM = \P'$, $\NN^T = \X_{\AA}$
\end{Frame}
In this section, we give a sublinear time algorithm which relies on constructing \textit{column \textrm{and} row pcps}, which in turn rely on our column and row norm estimators. Intuitively, we begin with obtaining coarse estimates to column norms. Next, we sample a subset of the columns of $\AA$ with probability proportional to their column norm estimates to obtain a \textit{ column pcp} for $\AA$. We show that the rescaled matrix still has enough structure to get a coarse estimate to its row norms. Then, we compute the row norm estimates of the sampled rescaled matrix, and subsample its rows to obtain a small matrix that is a \textit{row pcp}. We run an input-sparsity time algorithm (\cite{clarkson2013low}) on the small matrix to obtain a low-rank approximation. The main theorem we prove is as follows: 
%Note, the low-rank approximation of the smaller matrix is not a low-rank approximation for the original. In fact, it does not even have the right dimensions. Therefore, we follow a filtering step \cite{fkv04} to output a rank-k matrix of the right dimensions. We describe each of these steps in detail below. 
\begin{theorem}(Sublinear Low-Rank Approximation.)
\label{thm:first_sublinear}
Let $\AA \in \mathbb{R}^{m \times n}$ be a matrix that satisfies approximate triangle inequality. Then, for any $\epsilon > 0$ and integer $k$, Algorithm \ref{alg:first_sublinear} runs in time $O\left( \left(m^{1.34} + n^{1.34}\right)\text{poly}(\frac{k}{\epsilon}) \right)$  and outputs matrices $\MM \in \mathbb{R}^{m \times k}$, $\NN \in \mathbb{R}^{n \times k}$ such that with probability at least $9/10$,
\begin{equation*}
\left\| \AA - \MM \NN^T \right\|^2_{F} \leq  \left\| \AA - \AA_k\right\|^2_{F} + \epsilon \left\| \AA \right\|^2_F 
\end{equation*}
\end{theorem}

%\begin{theorem}(Frieze-Kannan-Vempala Row Sampling)
%\label{thm:fkv_row_sampling}
%Let $\AA$ be an $m \times n$ matrix and $P = \{p_1, p_2 \ldots p_m \}$ be a probability distribution over the rows of $\AA$ such that  $p_i \geq c \frac{\left\| \AA_{i,*}\right\|^2_2}{\left\| \AA \right\|^2_F}$. Let $\T\AA$ be a scaled sample of $s_1$ rows of $\AA$ according to the distribution $P$. Let $V$ be the vector space spanned by the columns of $\T\AA$. Then, for any $\epsilon > 0$,  with probability at least $9/10$, there exists an orthonormal set of $k$ vectors, $\{ y_1, y_2 \ldots y_k\} \in V$ such that 
%\[
%\| \AA - \AA(\sum_{i} y_i y_i^T) \|^2_F \leq \left\| \AA - \AA_k \right\|^2_F + \frac{10\epsilon}{cs_1} \left\| \AA \right\|^2_F
%\]
%\end{theorem}

%Observe, by Lemma $\ref{lem:row_norm_estimation}$, we can efficiently compute a distribution $p$ such that $p_i \geq \Theta\left(\frac{b_1}{n}\right)\frac{\left\| \AA_{*,j}\right\|^2_2}{\left\| \AA \right\|^2_F}$. Thus $c = \Theta\left(\frac{b_1}{n}\right)$. Sampling $s_1 = \Theta\left(\frac{n}{b_1} \text{poly}(\frac{k}{\epsilon})\right)$ rows of $\AA$ spans a vector space that contains a good low-rank approximation to $\AA$ satisfying the guarantees of Theorem $\ref{thm:fkv_row_sampling}$. Let $\T\AA$ be the $s_1\times n$ matrix formed by the sampling algorithm. We note that since we can estimate row norms to a $O\left(\frac{n}{b_1}\right)$-factor, and thus sample $s_1$ rows of $A$ to satisfy the guarantee as Theorem $\ref{thm:fkv_row_sampling}$. We refer to this algorithm as Frieze-Kannan-Vempala Row Sampling. 

\paragraph{Column Sampling.}
We observe that constructing a column and row projection-cost preserving sketches require sampling columns  proportional to their relative norms and subsequently subsample columns proportional to the relative row norms of the sampled, rescaled matrix. In the previous section, we obtain coarse approximations to these norms and thus use our estimates to serve as a proxy for the real distribution. For $j \in [n]$, let $\widetilde{X}_j$ be our estimate for the column $\AA_{*,j}$. We define a probability distribution over the columns of $\AA$ as $q_j = \frac{\widetilde{X}_j}{\sum_{i^{'}\in[n]}\widetilde{X}_{j^{'}}}$. Given that we can estimate column norms up to an $O\left(\frac{m}{b_1}\right)$-factor, $q_j \geq \left(\frac{b_1}{m}\right)\frac{\left\| \AA_{*,j}\right\|^2_2}{\left\| \AA \right\|^2_F}$, where $b_1$ is a parameter to be set later. Therefore, we oversample columns of $\AA$ by a $\Theta\left(\frac{m}{b_1}\right)$-factor to construct a \textit{column pcp} for $\AA$. Let $\AA\S$ be a scaled sample of $s_1=\Theta\left( \frac{mk^2\log(m)}{b_1\epsilon^2} \right)$ columns of $\AA$ such that each column is set to $\frac{\AA_{*,j}}{\sqrt{s_1q_j}}$ with probability $q_j$. Then, by Theorem \ref{thm:pcp_main_col} for any $\epsilon > 0$, with probability at least $99/100$, for all rank-$k$ projection matrices $\X$, 
\begin{equation}
\| \AA\S - \X \AA\S \|^2_F \leq \left\| \AA - \X\AA_k \right\|^2_F + \epsilon \left\| \AA \right\|^2_F
\end{equation}
We observe that the total time taken to construct a \textit{column pcp}  is $O\left(\frac{m\log(m)}{b_1}\text{poly}(\frac{k}{\epsilon}) + b_1n + m \right)$ and the total number of entries of $\AA$ queried are $O(b_1n + m)$.

\paragraph{Handling Rescaled Columns.}
We note that during the construction of the \textit{column pcp}, the $j^{th}$ column, if sampled, is rescaled by $\frac{1}{\sqrt{s_1 q_j}}$. Therefore, the resulting matrix, $\AA\S$ may no longer be a distance matrix. To address this issue, we partition the columns of  $\AA\S$ into weight classes such that the $g^{th}$ weight class contains column index $j$ if the corresponding scaling factor $\frac{1}{\sqrt{ q_j}}$ lies in the interval  $\left[(1+\epsilon)^g, (1+\epsilon)^{g+1}\right)$. Note, we can ignore the $(\frac{1}{\sqrt{s_1}})$-factor since every entry is rescaled by the same constant. Formally, 
\begin{equation}
\mathcal{W}_g = \left\{ i \in [s_1] \textrm{ $\Big{|}$ } \frac{1}{\sqrt{q_j}} \in \left[(1+\epsilon)^g, (1+\epsilon)^{g+1}\right) \right\}
\end{equation}
Next, with high probability, for all $j\in[n]$, if column $j$ is sampled, $\frac{1}{q_j} \leq n^c$ for a large constant $c$. If instead, $q_j \leq \frac{1}{n^{c'}}$, the probability that the $j^{th}$ is sampled would be at most $1/n^{c'}$, for some $c' > c$. Union bounding over such events for $n$ columns, the number of weight classes is at most $\log_{1+\epsilon}(n^c) = $     $O\left(\epsilon^{-1}\log(n)\right)$.
Let $\AA\S_{|\mathcal{W}_g}$ denote the columns of $\AA\S$ restricted to the set of indices in $\mathcal{W}_g$. Observe that all entries in $\AA\S_{|\mathcal{W}_g}$ are scaled to within a $(1+\epsilon)$-factor of each other and therefore, satisfy approximate triangle inequality (equation \ref{eqn:approx_tri}).
Therefore, row norms of $\AA\S_{|\mathcal{W}_g}$ can be computed using Algorithm \ref{alg:row_norm_estimation} and the estimator is an $O\left(\frac{n}{b_2}\right)$-approximation (for some parameter $b_2$), since Lemma \ref{lem:row_norm_estimation} blows up by a factor of at most $1+\epsilon$. Summing over the estimates from each partition above, with probability at least $99/100$, we obtain an $O\left(\frac{n}{b_2}\right)$-approximate estimate to $\|\AA\S_{i,*}\|^2_2$, simultaneously for all $i \in [m]$. However, we note that each iteration of Algorithm \ref{alg:row_norm_estimation} reads $b_2m + n$ entries of $\AA$ and there are at most $O(\epsilon^{-1}\log(n))$ iterations. Therefore, the time taken to compute the estimates to the row norms is $O\left( (b_2m + n)\epsilon^{-1}\log(n) \right)$.

\paragraph{Row Sampling.}
Next, we construct a \textit{row pcp} for $\AA\S$. For $i \in [m]$, let $\widetilde{X}_i$ be an $O\left(\frac{n}{b_2}\right)$-approximate estimate for $\|\AA\S_{i,*}\|^2_2$. Let $p = \{p_1, p_2, \ldots, p_m \}$ be a distribution over the rows of $\AA\S$ such that $p_i= \frac{\widetilde{X}_i}{\sum_i \widetilde{X}_i} \geq \left(\frac{b_2}{n} \frac{\|\AA\S_{i,*}\|^2_2}{\|\AA\S \|^2_F} \right)$. Therefore, we oversample rows by a $\Theta\left(\frac{n}{b_2}\right)$ factor to obtain a \textit{row pcp} for $\AA\S$. Let $\T\AA\S$ be a scaled sample of $s_2 = \Theta \left( \frac{nk^2\log(n)}{b_2\epsilon^2} \right)$ rows of $\AA\S$ such that each row is set to $\frac{\AA\S_{i,*}}{\sqrt{s_2 p_i}}$ with probability $p_i$. By Corollary \ref{thm:pcp_main_row}, with probability at least $99/100$, for all rank-$k$ projection matrices $\X$, 
\begin{equation}
\|\T\AA\S - \T\AA\S\X \|^2_F \leq \| \AA\S - \AA\S\X \|^2_F + \epsilon\| \AA\S\|^2_F 
\end{equation}
We observe that the total time taken to construct a \textit{row pcp} is $O\left( \frac{n \log(n)}{b_2}\textrm{poly}(\frac{k}{\epsilon}) + (b_2m + n )\frac{\log(n)}{\epsilon} \right)$ and the total number of entries of $\AA$ queried is $O\left((b_2m + n )\frac{\log(n)}{\epsilon} \right)$.

\paragraph{Input-sparsity Time Low-Rank Approximation.}
Next, we compute a low-rank approximation for the smaller matrix, $\T\AA\S$, in input-sparsity time. To this end we use the following theorem from \cite{clarkson2013low}:

\begin{theorem}(Clarkson-Woodruff LRA.)
\label{thm:clarkson_woodruff}
For $\AA \in \mathbb{R}^{m \times n}$, there is an algorithm that with failure probability at most $1/10$ finds $L \in R^{m \times k}$, $W \in R^{n \times k}$ and a diagonal matrix $D\in R^{k \times k}$, such that \[\left\| \AA - \LL \mathbf{D} \mathbf{W}^T \right\|^2_F \leq (1 + \epsilon) \left\| \AA - \AA_k \right\|^2_F\] 
and runs in time $O\left(\texttt{nnz}(\AA) + (n+m)\text{poly}(\frac{k}{\epsilon})\right)$, where $\texttt{nnz}(\AA)$ is the number of non-zero entries in $\AA$. 
\end{theorem}
Running the input-sparsity time algorithm with the above guarantee on the matrix $\T\AA\S$, we obtain a rank-$k$ matrix $\LL \mathbf{D} \mathbf{W}^T$, such that
\begin{equation}
\| \T\AA\S - \LL \mathbf{D} \mathbf{W}^T\|^2_F \leq (1+\epsilon) \|\T\AA\S - (\T\AA\S)_k \|^2_F
\end{equation}
where $(\T\AA\S)_k$ is the best rank-$k$ approximation to $\T\AA\S$ under the Frobenius norm. Since $\T\AA\S$ is a small matrix, we can afford to read all of it by querying at most $O\left(\frac{nm\log(n)\log(m)}{b_1b_2}\textrm{poly}(\frac{k}{\epsilon}) \right)$ entries of $\AA$ and the algorithm runs in time $O\left( s_1s_2 + (s_1 + s_2)\textrm{poly}(\frac{k}{\epsilon} ) \right) $.

\paragraph{Constructing a solution for $\AA$.} Note, while $\LL \mathbf{D} \mathbf{W}^T$ is an approximate rank-$k$ solution for $\T\AA\S$, it does not have the right dimensions as $\AA$. If we do not consider running time, we could construct a low-rank approximation to $\AA$ as follows: since projecting $\T\AA\S$ onto $\WW^T$ is approximately optimal, it follows from Lemma \ref{lem:pcp_useful} that with probability $98/100$,
\begin{equation}
\label{eqn:pcp_as}
\| \AA\S -\AA\S\WW\WW^T  \|^2_F =\| \AA\S - (\AA\S)_k \|^2_F \pm \epsilon \| \AA\S \|^2_F
\end{equation}
Let $(\AA\S)_k = \P\mathbf{N}^T$ be such that $\P$ has orthonormal columns. Then, $\|\AA\S - \P\P^T \AA\S \|^2_F = \| \AA\S - (\AA\S)_k \|^2_F$ and by Lemma \ref{lem:pcp_useful} it follows that with probability $98/100$, $\|\AA - \P\P^T\AA\|^2_F \leq \|\AA - \AA_k \|^2_F + \epsilon\| \AA \|_F$. However, even approximately computing a column space $\P$ for $(\AA\S)_k$ using an input-sparsity time algorithm is no longer sublinear. To get around this issue, we observe that an approximate solution for $\T\AA\S$ lies in the row space of $\WW^T$ and therefore, an approximately optimal solution for $\AA\S$ lies in the row space of $\WW^T$. We then set up the following regression problem
\begin{equation}
\min_{\X} \|\AA\S - \X \WW^T \|^2_F
\end{equation}
Note, this regression problem is still too big to be solved in sublinear time. Therefore, we sketch it by sampling columns of $\AA\S$ according to the leverage scores of $\WW^T$ to set up a smaller regression problem. Formally, we use a theorem of \cite{clarkson2013low} (Theorem 38) to approximately solve this regression problem (also see \cite{drineas2008relative} for previous work.)  
\begin{theorem}(Fast Regression.)
\label{thm:clarkson_woodruff_regression}
Given a matrix $\AA \in \mathbb{R}^{m \times n}$ and a rank-$k$ matrix $\BB \in \mathbb{R}^{m \times k}$, such that $\BB$ has orthonormal columns, the regression problem $\min_{\X} \|\AA - \BB\X\|^2_F$ can be solved up to $(1+\epsilon)$ relative error, with probability at least $2/3$ in time $O\left( (m\log(m) + n)\textrm{poly}(\frac{k}{\epsilon}) \right)$ by constructing a sketch $\E$ with poly$(\frac{k}{\epsilon})$ rows and solving $\min_{\X}\|\E\AA - \E\BB\X\|^2_F$. Note, a similar guarantee holds for solving $\min_{\X}\|\AA - \X\BB\|^2_F$.
\end{theorem}
Since $\WW^T$ has orthonomal rows, the leverage scores are precomputed. With probability at least $99/100$, we can compute $\X_{\AA\S} = \argmin_{\X} \|\AA\S\E - \X \WW^T\E \|^2_F$, where $\E$ is a leverage score sketching matrix with poly$\left(\frac{k}{\epsilon}\right)$ columns, as shown in Theorem \ref{thm:clarkson_woodruff_regression}.  
\begin{equation}
\begin{split}
\|\AA\S - \X_{\AA\S} \WW^T \|^2_F & \leq (1+\epsilon) \min_{\X}\|\AA\S - \X \WW^T \|^2_F \\
& \leq (1+\epsilon) \| \AA\S -\AA\S\WW\WW^T  \|^2_F \\ 
& = \| \AA\S - (\AA\S)_k \|^2_F \pm \epsilon \| \AA\S \|^2_F
\end{split}
\end{equation}
where the last two inequalities follow from equation \ref{eqn:pcp_as}. Recall, $\AA\S$ is an $m \times s_1$ matrix and thus the running time is $O\left((m +  s_1) \log(m) \textrm{poly}\left(\frac{k}{\epsilon} \right) \right)$.
Let  $\X_{\AA\S}\WW^T = \P'\mathbf{N}'^T $ be such that $\P'$ has orthonormal columns. Then, the column space of $\P'$ contains an approximately optimal solution for $\AA$, since $\|\AA\S - \P'\mathbf{N}'^T \|^2_F = \| \AA\S - (\AA\S)_k \|^2_F \pm \epsilon \|\AA\S \|^2_F$ and $\AA\S$ is a \emph{column pcp} for $\AA$. It follows from Lemma \ref{lem:pcp_useful}  that with probability at least $98/100$, 
\begin{equation}
\label{eqn:additive_for_a}
\| \AA - \P'\P'^T\AA \|^2_F \leq \| \AA - \AA_k \|_F + \epsilon \|\AA\|_F
\end{equation}
Therefore, there exists a good solution for $\AA$ in the column space of $\P'$. Since we cannot compute this explicitly, we set up the following regression problem: 
\begin{equation}
\min_{\X} \|\AA - \P'\X \|^2_F
\end{equation}
Again, we sketch the regression problem above by sampling columns of $\AA$ according to the leverage scores of $\P'$. We can then compute $\X_{\AA} = \argmin_{\X} \| \E' \AA - \E' \P'\X \|^2_F $ with probability at least $99/100$, where $\E'$ is a leverage score sketching matrix with poly$\left(\frac{k}{\epsilon}\right)$ rows. Then, using the properties of leverage score sampling from Theorem \ref{thm:clarkson_woodruff_regression},
\begin{equation}
\begin{split}
\| \AA -  \P'\X_{\AA} \|^2_F & \leq (1+\epsilon)\min_{\X} \| \AA - \P'\X \|^2_F \\
& \leq (1+\epsilon) \|\AA - \P' \P'^T\AA \|^2_F \\
& \leq  \| \AA - \AA_k \|^2_F + O(\epsilon) \|\AA\|^2_F
\end{split}
\end{equation}
where the second inequality follows from $\X$ being the minimizer and $\P'^T\AA$ being some other matrix, and the last inequality follows from equation \ref{eqn:additive_for_a}.
Recall, $\P'$ is an $m \times k$ matrix and by Theorem 38 of CW, the time taken to solve the regression problem is $O\left( (m\log(m) + n) \textrm{poly}\left(\frac{k}{\epsilon}\right) \right)$.  Therefore, we observe that $\P'\X_{\AA}$ suffices and we output it in factored form by setting $\MM = \P'$ and $\NN = \X_{\AA}^T $. Union bounding over the probabilistic events, and rescaling $\epsilon$, with probability at least $9/10$, Algorithm \ref{alg:first_sublinear} outputs $\MM \in \mathbf{R}^{m \times k}$ and $\NN \in \mathbf{R}^{n \times k}$ such that the guarantees of Theorem \ref{thm:first_sublinear} are satisfied. 

Finally, we analyze the overall running time of Algorithm \ref{alg:first_sublinear}. Computing the estimates for the column norms and constructing the \textit{column pcp} for $\AA$  has running time $O\left(\frac{m\log(m)}{b_1}\text{poly}(\frac{k}{\epsilon}) + b_1n + m \right)$. Then, computing estimates for row norms and constructing a \textit{row pcp} for $\AA\S$ has overall running time $O\left( \frac{n \log(n)}{b_2}\textrm{poly}(\frac{k}{\epsilon}) + (b_2m + n )\frac{\log(n)}{\epsilon} \right)$. The input-sparsity time algorithm to compute a low-rank approximation of $\T\AA\S$ has running time $O\left( s_1s_2 + (s_1 + s_2)\textrm{poly}(\frac{k}{\epsilon} ) \right)$ and constructing a solution for $\AA$ is dominated by $O\left( (m\log(m) + n) \textrm{poly}\left(\frac{k}{\epsilon}\right) \right)$. Therefore, the overall running time is dominated by $O\Big((b_1 n + b_2 m )\frac{\log(n)}{\epsilon} + \frac{mn \log(m) \log(n)}{b_1 b_2}\textrm{poly}\left( \frac{k}{\epsilon} \right) + \left(\frac{n\log(n)}{b_2} + \frac{m\log(m)}{b_1} \right)\textrm{poly}\left( \frac{k}{\epsilon} \right)  \Big)$.
Setting $b_1 =\frac{\epsilon n^{0.34}}{\log(n)} $ and $b_2 = \frac{\epsilon m^{0.34}}{\log(m)} $, we note that the overall running time is 
$\widetilde{O}\left( ( m^{1.34} + n^{1.34})\textrm{poly}\left( \frac{k}{\epsilon}\right)  \right)$
where $\widetilde{O}$ hides $\log(m)$ and $\log(n)$ factors. This completes the proof of Theorem \ref{thm:first_sublinear}. 
%Following the filter step of Frieze-Kannan-Vempala, let $T = \left\{t \textrm{ $|$ } \|(\T\AA\S)^T \ell_t\|^2_2 \geq \gamma \|\T\AA \|^2_F \right\}$, where $\gamma = \frac{c\epsilon}{8k}$. For all $t\in T$, let $v_t = \frac{\AA^T\T^T \ell_t}{\|(\T\AA\S)^T \ell_t \|_2}$. We then project $\AA$ on to $\sum_{t \in T} v_t v_t^T$, which is a $m \times n$ with rank at most $k$. Following the Frieze-Kannan-Vempala analysis, 
%\begin{equation}
%\| \AA - \AA \sum_{t \in T} v_t v_t^T \|^2_F \leq (1+\epsilon)\| \AA - \AA_k \|^2_F + \epsilon \| \AA \|^2_F \leq  \| \AA - \AA_k \|^2_F + O(\epsilon)\| \AA \|^2_F
%\end{equation} 
%Recall, Clarkson-Woodruff runs in time in $O\left(nnz(\mathbf{\T\AA\S}) + (m+n)\text{poly}(\frac{k}{\epsilon}))\right)$. We observe that $nnz(\T\AA\S) = O\left(\frac{m n}{b_1 b_2}\text{poly}(\frac{k}{\epsilon})\right)$. 
%Therefore, the total time taken by estimation, sampling and low-rank approximation is $O\left(\left(\frac{mn}{b_1 b_2} + \frac{b_2n + b_1m}{b_1 b_2}\right)\text{poly}(\frac{k}{\epsilon}) + b_1m + b_2n\frac{\log(m)}{\epsilon} \right)$. We observe that this immediately implies a sublinear time algorithm by setting $b_1 = m^{0.34}$ and $b_2 = \frac{\epsilon n^{0.34}}{\log(m)}$. 

%% file: pcp.tex
\section{Projection-Cost Preserving Sketches}

Next, we describe how to use the above estimators to sample rows and columns. We would like to reduce the dimensionality of the matrix in a way that approximately preserves low-rank structure. At a high level, we sketch the input matrix on the left and the right and use an input-sparsity time algorithm on the resulting smaller matrix. The main insight to show such a result is that if we can approximately preserve all rank-$k$ subspaces in the column and row space of the matrix, then we can recursively sample rows and columns to obtain a much smaller matrix. To this end, we introduce a relaxation of projection-cost preserving sketches \cite{cmm17} that satisfy an {\it additive error guarantee}.

We prove that projection-cost preserving sketches can be computed using our coarse estimates to the column and row norms. %Since their introduction, projection-cost preserving sketches have found widespread applications in machine learning problems such as clustering. We hope that the additive error guarantee we introduce can be applied to an even wider range of problems. Further, our  
We begin by showing an intermediate result bounding the spectrum of the sample in terms of the original matrix. 

\begin{theorem}(Spectral Bounds.)\label{thm:pcp1}
Let $\AA$ be a $m \times n$ matrix such that $\AA$ satisfies approximate triangle inequality. For $j \in [n]$, let $\widetilde{X}_j$ be an $O\left(\frac{m}{b}\right)$-approximate estimate for the $j^{th}$ column of $\AA$ such that it satisfies the guarantee of Corollary \ref{cor:column_norm_estimation}. Then, let $q = \{q_1, q_2 \ldots q_n\}$ be a probability distribution over the columns of $\AA$ such that $q_j = \frac{\widetilde{X}_j}{\sum_{j'}\widetilde{X}_{j^{'}}}.$ Let $t =  O\left(\frac{ m}{b\epsilon^2}\log(\frac{m}{\delta})\right)$ for some constant $c$. Then, construct $\CC$ using $t$ columns of $\AA$ and set each one to $\frac{A_{*,j}}{\sqrt{t q_j}}$ with probability $q_j$. With probability at least $1-\delta$,
\[
\CC\CC^T - \epsilon\| \AA \|^2_F \II \preceq \AA\AA^T \preceq \CC\CC^T + \epsilon\|\AA\|^2_F \II
\]
\end{theorem}
%=  O\left(\frac{ m}{b\epsilon^2}\log(\frac{m}{\delta})\right)
\begin{proof}
Let $\Y = \CC\CC^T - \AA\AA^T$. For notational convenience let $\AA_j = \AA_{*,j}$. We can then write $\Y = \sum_{i\in [t]} \X_i$, where $\X_i = \frac{1}{t}(\frac{1}{q_j}\AA_{j}\AA^T_{j} - \AA\AA^T )$ with probability $q_j$. We observe that $\expec{}{\frac{1}{q_j}\AA_{j}\AA^T_{j} - \AA\AA^T}= 0$, and therefore, $\expec{}{\Y} =0$. Next, we bound the operator norm of $\Y$. To this end, we use the Matrix Bernstein inequality, which in turn requires a bound on the operator norm of $\X_i$ and variance of $\Y$. 
Recall, 
\begin{equation}
    \begin{split}
        \| \X_i \|_2 & = \left\| \frac{1}{tq_j}\AA_{j}\AA^T_{j} - \frac{1}{t}\AA\AA^T \right\|_2 \\
        & \leq \frac{m}{tb} \frac{\| \AA \|^2_F}{\| \AA_j \|^2_2}\|\AA_{j}\AA^T_{j} \|_2 + \frac{1}{t} \| \AA\AA^T \|_2 \\
        & \leq \frac{2m}{tb} \| \AA \|^2_F 
    \end{split}
\end{equation}
Next, we bound $\var{\Y}\leq \expecf{}{\Y^2}$.
\begin{equation}
    \begin{split}
        \expecf{}{\Y^2} & = t \expecf{}{\X^2_i} = \frac{1}{t}\expecf{}{\left(\frac{1}{q_j}\AA_{j}\AA^T_{j} - \AA\AA^T\right)^2} \\
        & = \frac{1}{t}\left( \frac{(\AA_{j}\AA^T_{j})^2}{q_j} + q_j(\AA\AA^T)^2 -2 \AA_{j}\AA^T_{j} \AA\AA^T \right) \\
        & \preceq \frac{1}{t}\left( \frac{m}{b}\frac{\| \AA \|^2_F}{\| \AA_j \|^2_2}(\AA_{j}\AA^T_{j})^2 + \frac{b}{m}\frac{\| \AA_j \|^2_2}{\| \AA \|^2_F}(\AA\AA^T)^2 \right) \\
        & \preceq \frac{cm\| \AA \|^4_F }{tb}\II_{m\times m} 
    \end{split}
\end{equation}
Therefore, $\sigma^2 = \| \expecf{}{\Y^2} \|_2 \leq \frac{cm\| \AA \|^4_F }{tb}$ Applying the Matrix Bernstein inequality (see Lemma \ref{lem:matrix_bern} in the Appendix), 
\[
\prob{}{\| \Y \|_2 \geq \epsilon\| \AA \|^2_F } \leq m e^{\left(- \frac{\epsilon^2 \| \AA \|^4_F}{\sigma^2 + \frac{\epsilon m}{3tb}\| \AA \|^4_F }\right)}
\leq \delta
\]
by substituting the value of $\sigma^2$ and setting $t = \frac{c' m}{b\epsilon^2}\log(\frac{m}{\delta})$. The bound follows. 
\end{proof}

Using the above theorem, we show that sampling columns of $\AA$ according to approximate column norms yields a matrix that preserves projection cost for all rank-$k$ projections up to additive error. 
\begin{theorem}(Column Projection-Cost Preservation.)
\label{thm:pcp_main_col}
Let $\AA$ be a $m \times n$ matrix such that $\AA$ satisfies approximate triangle inequality. For $j \in [n]$, let $\widetilde{X}_j$ be an $O\left( \frac{m}{b} \right)$-approximate estimate for the $j^{th}$ column of $\AA$ such that it satisfies the guarantee of Corollary \ref{cor:column_norm_estimation}. Then, let $q = \{q_1, q_2 \ldots q_n\}$ be a probability distribution over the columns of $\AA$ such that $q_j = \frac{\widetilde{X}_j}{\sum_{j'}\widetilde{X}_{j'}}$. Let $t = O\left( \frac{ m k^2}{b\epsilon^2}\log(\frac{m}{\delta})\right)$ for some constant $c$. Then, construct $\CC$ using $t$ columns of $\AA$ and set each one to $\frac{A_{*,j}}{\sqrt{t q_j}}$ with probability $q_j$. With probability at least $1-\delta$, for any rank-$k$ orthogonal projection $\X$
\[
\| \CC - \X\CC \|^2_F = \| \AA - \X\AA \|^2_F \pm \epsilon\| \AA \|^2_F
\]
\end{theorem}
%$O\left( \frac{ m k^2}{b\epsilon^2}\log(\frac{m}{\delta})\right)
\begin{proof}
We give a similar proof to the relative error guarantees in \cite{cmm17}, but need to replace certain parts with our different distribution which is only based on row and column norms rather than leverage scores, and consequently we obtain additive error in places instead. As our lower bound shows, this is necessary in our setting.

Let $\Y = \mathbb{I}-\X$, so that $\| \AA - \X\AA \|^2_F = \trace{\Y \AA \AA^T \Y}$ and $\| \CC - \X\CC \|^2_F = \trace{\Y \CC \CC^T \Y}$.
We split the singular values of $\AA$ into a head and a tail as follows. Let $\sigma^2_\ell$ be the smallest singular value of $\AA$ for which $\sigma^2_\ell \geq \frac{\| \AA \|^2_F}{k}$. Let $\U_\ell \U_\ell^T$ be the projection onto the top $\ell$ singular vectors of $\AA$  and $\U_{\setminus \ell} \U_{\setminus \ell}^T$ be the projection on the bottom singular vectors. Let $\P_\ell = \U_\ell \U_\ell^T$ and $\P_{\setminus \ell} = \U_{\setminus \ell} \U_{\setminus \ell}^T$. Then, 
\begin{equation} \label{eqn:split1}
\begin{split}
 \trace{\Y \AA \AA^T \Y}  &= \trace{\Y \P_\ell \AA \AA^T \P_\ell \Y}
  + \trace{\Y \P_{\setminus \ell} \AA \AA^T \P_{\setminus \ell} \Y} + 2 \trace{\Y \P_\ell \AA \AA^T \P_{\setminus \ell} \Y} \\
 & = \trace{\Y \P_\ell \AA \AA^T \P_\ell \Y} + \trace{\Y \P_{\setminus \ell} \AA \AA^T \P_{\setminus \ell} \Y}
\end{split}
\end{equation}
The cross terms vanish since $\P_\ell \AA$ and $\P_{\setminus \ell}\AA$ are orthogonal. Similarly, we split the $\CC \CC^T$ terms.
\begin{equation} \label{eqn:split2}
\begin{split}
 \trace{\Y \CC \CC^T \Y} &= \trace{\Y \P_\ell \CC \CC^T \P_\ell \Y}
  + \trace{\Y \P_{\setminus \ell} \CC \CC^T \P_{\setminus \ell} \Y} + 2 \trace{\Y \P_\ell \CC \CC^T \P_{\setminus \ell} \Y} \\
\end{split}
\end{equation}
We note that the cross terms here do not vanish since $\P_\ell \CC$ and $\P_{\setminus \ell} \CC$ might not be orthogonal. We now show how to handle each of these terms separately. 
\subsection{Head Terms.} 
Our analysis for the Head Terms closely follows that of \cite{cmm17}, where they strive for relative error using leverage scores instead.  
For any vector $x$, let $y = \P_\ell x$. Then, $y^T \AA \AA^T y = x^T \P_\ell^T \AA \AA^T \P_\ell x = x^T \AA_\ell \AA_\ell^T x $. Then, setting $\epsilon = \frac{\epsilon}{k}$ in Theorem $\ref{thm:pcp1}$, we obtain
\begin{equation}
\label{eqn:head1}
    y^T \CC \CC^T y - \frac{\epsilon \| \AA \|^2_F}{k}y^T y \leq x^T \AA_\ell \AA^T_\ell x \leq y^T \CC \CC^T y + \frac{\epsilon \| \AA \|^2_F}{k}y^T y
\end{equation}
Note, this only increases the number of columns we sample by a poly$(k)$ factor. Recall, by definition, $y$ is orthogonal to all but the top $\ell$ singular vectors of $\AA$. Therefore, $x^T \AA_\ell \AA_\ell^T x = y^T \AA \AA^T y \geq \frac{\| \AA \|^2_F}{k}y^Ty$. Combined with $(\ref{eqn:head1})$, $y^T \CC \CC^T y = (1 \pm \epsilon)x^T \AA_\ell \AA^T_\ell x$. Since $y^T \CC \CC^T y = x^T \P_\ell \CC \CC^T \P_\ell x$ and the above is true for any $x$, we get 
\begin{equation}
\label{eqn:head2}
    (1 - \epsilon) \P_\ell \CC \CC^T \P_\ell \preceq \AA_\ell \AA_\ell^T \preceq (1+\epsilon) \P_\ell \CC\CC^T \P_\ell
\end{equation}
We observe that $(\ref{eqn:head2})$ bounds the diagonal entries of $\Y \AA_\ell \AA_\ell^T \Y$ in terms of the diagonal entries of $\Y \P_\ell \CC \CC^T \P_\ell \Y$, we get that, 
$$
(1 - \epsilon) \trace{\Y \P_\ell \CC \CC^T \P_\ell \Y}\leq \trace{\Y \P_\ell \AA \AA^T \P_\ell \Y} \leq (1 + \epsilon) \trace{\Y \P_\ell \CC \CC^T \P_\ell \Y}
$$
Rearranging the terms and assuming $\epsilon < 1/2$,
\begin{equation}
\label{eqn:head_main}
(1 - 4\epsilon) \trace{\Y \P_\ell \AA \AA^T \P_\ell \Y} \leq \trace{\Y \P_\ell \CC \CC^T \P_\ell \Y} \leq (1 + 4\epsilon) \trace{\Y \P_\ell \AA \AA^T \P_\ell \Y}
\end{equation}

\subsection{Tail Terms.}
Recall from the definition of $\Y$, 
\begin{equation}
    \begin{split}
    \trace{\Y \AA_{\setminus \ell} \AA_{\setminus \ell}^T\Y} & =\trace{\AA_{\setminus \ell} \AA_{\setminus \ell}^T} 
     -\trace{\X \AA_{\setminus \ell} \AA_{\setminus \ell}^T\X}
    \end{split}
\end{equation}
Similarly, 
\begin{equation}
    \begin{split}
    \trace{\Y \P_{\setminus \ell} \CC \CC^T \P_{\setminus \ell} \Y} & = \trace{\P_{\setminus \ell} \CC \CC^T \P_{\setminus \ell} } 
     - \trace{\X \P_{\setminus \ell} \CC \CC^T \P_{\setminus \ell} \X}
    \end{split}
\end{equation}
In order to relate $\trace{\P_{\setminus \ell} \CC \CC^T \P_{\setminus \ell}}$
and $\trace{\AA_{\setminus \ell} \AA_{\setminus \ell}^T}$, we observe that the first term is equal to the second in expectation. Therefore, using a scalar Chernoff bound, we show that $\trace{\P_{\setminus \ell} \CC \CC^T \P_{\setminus \ell}}$ concentrates around its expectation. We defer this proof to the Supplementary Material and obtain the following bound: 
\begin{equation}
\label{eqn:chernoff}
    \trace{\AA_{\setminus \ell} \AA_{\setminus \ell}^T} - \trace{\P_{\setminus \ell} \CC \CC^T \P_{\setminus \ell}} = \pm \epsilon \| \AA \|^2_F
\end{equation}
Next, we relate the remaining two terms following a strategy similar to the one used for the head terms. Let the vectors $x,y$ be defined as above. Then, $x^T \AA_{\setminus \ell} \AA^T_{\setminus \ell} x = y^T \AA \AA^T y$. Using Theorem \ref{thm:pcp1} with $\epsilon = \frac{\epsilon}{k}$, we get 
\begin{equation}
\label{eqn:tail1}
    x^T \AA_{\setminus \ell} \AA^T_{\setminus \ell} x = y^T \CC \CC^T y \pm \frac{\epsilon \| \AA \|^2_F}{k}y^T y
\end{equation}
Since $\P_{\setminus \ell}$ is a projection matrix, $y^Ty \leq x^Tx$ and assuming $\epsilon< 1/2$,
\begin{equation}
\label{eqn:tail2}
    \begin{split}
        y^T \CC \CC^T y  - \epsilon\frac{\| A\|^2_F}{k}x^Tx \leq x^T \AA_{\setminus \ell} \AA^T_{\setminus \ell}x \\
        x^T \AA_{\setminus \ell} \AA^T_{\setminus \ell}x \leq  y^T \CC \CC^T y + \epsilon\frac{\| \AA\|^2_F}{k}x^Tx 
    \end{split}
\end{equation}

Recall, by the definition of $m$, $x^T \AA_{\setminus \ell} \AA^T_{\setminus \ell}x \leq \frac{\| A\|^2_F}{k}$. Substituting this back into $(\ref{eqn:tail2})$ and $(\ref{eqn:tail3})$, we get
\begin{equation}
\label{eqn:tail3}
    P_{\setminus \ell} \CC \CC^T \P_{\setminus \ell} - \epsilon\frac{\| \AA\|^2_F}{k}\II \preceq \AA_{\setminus \ell} \AA_{\setminus \ell}^T \preceq \P_{\setminus \ell} \CC \CC^T \P_{\setminus \ell} + \epsilon\frac{\| \AA\|^2_F}{k}\II
\end{equation}

Let $\X = \mathbf{ZZ}^T$ such that $\mathbf{Z} \in \mathbb{R}^{m\times k}$ is an orthonomal matrix. By the cyclic property of the trace,
\[
\trace{\X \AA_{\setminus \ell} \AA_{\setminus \ell}^T\X} = \trace{\mathbf{Z}^T \AA_{\setminus \ell} \AA_{\setminus \ell}^T\mathbf{Z}}
= \sum_{j\in[k]} \mathbf{Z}^T_{*,j} \AA_{\setminus \ell} \AA_{\setminus \ell}^T \mathbf{Z}_{*,j}
\]
Similarly, 
\[
\trace{\X \CC_{\setminus \ell} \CC_{\setminus \ell}^T\X} \leq \sum_{j\in[k]} \mathbf{Z}^T_{*,j} \CC_{\setminus \ell} \CC_{\setminus \ell}^T \mathbf{Z}_{*,j}
\]
Combining this with $(\ref{eqn:tail3})$ and $(\ref{eqn:tail1})$, and assuming $\epsilon < 1/2$ we get,  
\begin{equation}
\label{eqn:tail_final}
    \trace{\Y \P_{\setminus \ell} \CC \CC^T \P_{\setminus \ell} \Y} = \trace{\Y \AA_{\setminus \ell} \AA_{\setminus \ell}^T\Y} \pm 4\epsilon\| A \|^2_F
\end{equation}

\subsection{Cross Terms.}
Finally, we consider the cross term $2 \trace{\Y \P_\ell \CC \CC^T \P_{\setminus \ell} \Y}$. Let $\LL = \AA \AA^T (\AA \AA^T)^{+}$ and $\MM = \AA \AA^T$. We observe that the columns of $\P\CC\CC^T\P_{\setminus \ell}$ lie in the column span of $\AA$. Therefore,
\begin{equation}
\label{eqn:cross1}
    \trace{\Y \P_\ell \CC \CC^T \P_{\setminus \ell} \Y} = \trace{\Y \LL \P_\ell \CC \CC^T \P_{\setminus \ell} \Y}
\end{equation}
Then, by Cauchy-Schwarz, 
\begin{equation}
\label{eqn:cross2}
\begin{aligned}
    \trace{\Y \LL \P_\ell \CC \CC^T \P_{\setminus \ell} \Y} & \leq 
    \sqrt{\trace{\Y \LL \MM  \Y} \trace{\P_{\setminus \ell} \CC \CC^T \P_{\ell} \MM^{+} \P_\ell \CC \CC^T \P_{\setminus \ell} } } \\
    & = \sqrt{\trace{\Y \MM \Y}\trace{\P_{\setminus \ell}\CC\CC^T \U_\ell\Sigma^{-2}\U_\ell^T\CC\CC^T\P_{\setminus \ell}}} \\
    & = \sqrt{\trace{\Y \MM \Y} } \cdot \sqrt{\| \P_{\setminus \ell} \CC \CC^T \mathbf{U}_\ell \mathbf{\Sigma}^{-1}_\ell\|^2_F }
\end{aligned}
\end{equation}
Note, the first term is $\| \AA - \X\AA \|_F$. Therefore, we focus on the second term :
\begin{equation}
\label{eqn:cross3}
    \| \P_{\setminus \ell} \CC \CC^T \mathbf{U}_\ell \mathbf{\Sigma}^{-1}_\ell\|^2_F = \sum_{i\in[\ell]} \| \P_{\setminus \ell} \CC \CC^T \mathbf{U}_{i,*} \|^2_2\sigma^{-2}_i
\end{equation}
In order to bound the sum above, we bound each summand individually. Let $p_i$ be a unit vector in the direction of $\CC \CC^T \mathbf{U}_{i,*}$'s projection on $\P_{\setminus \ell}$. Then,
\begin{equation}
    \| \P_{\setminus \ell} \CC \CC^T \mathbf{U}_{i,*} \|^2_2 = (p_i^T \CC \CC^T \mathbf{U}_{i,*})^2 
\end{equation}
Let $\ell = \sigma^{-1}_i u_i +  \frac{\sqrt{k}}{\| \AA \|^2_F} p_i$. By Theorem \ref{thm:pcp1} we know that
\begin{equation}
    \ell^T \CC \CC^T\ell - \frac{\epsilon \| A\|^2_F}{k}\ell^T \ell \leq \ell^T \AA\AA^T\ell
\end{equation}
Substituting $\ell$ in the equation above,
\begin{equation}\label{eqn:cross4}
    \begin{split}
        \frac{\U_{i,*}\CC\CC^T\U_{i,*}}{\sigma^2_i} + \frac{kp^T_i\CC \CC^T p_i}{\| \AA \|^2_F} + \frac{2\sqrt{k}}{\|\AA\|_F}p_i^T\CC\CC^T\U_{i,*}
        & \leq \frac{\U_{i,*}\MM\U_{i,*}}{\sigma^2_i} + \frac{k}{\|\AA\|^2_F} + \frac{\epsilon\| \AA \|^2_F}{k}\ell^T\ell \\
        & = 1 + \frac{k}{\| \AA \|^2_F}p_i^T \AA \AA^T p_i + \frac{\epsilon\|\AA\|^2_F}{k}\ell^T\ell
    \end{split}
\end{equation}
Combining the above equation with (\ref{eqn:head2}) $\U_{i,*}^T\CC\CC^T\U_{i,*} \geq (1-\epsilon)\U_{i,*}^T\MM\U_{i,*} \geq (1-\epsilon)\sigma^2_i$. Further, $p_i^T \CC\CC^T p_i \geq p_i^T \MM p_i -\frac{\epsilon\| \AA \|^2_F}{k}$. Plugging this back into (\ref{eqn:cross4}), 
\begin{equation}
    \begin{split}
        & (1-\epsilon)\left(\frac{\U_{i,*}\MM\U_{i,*}}{\sigma^2_i} + \frac{kp_i^T \MM p_i}{\|\AA \|^2_F}\right)  + \frac{2\sqrt{k}p_i\CC\CC^Tu_i}{\sigma_i\| \AA \|_F}\\
        & \leq 1 + \frac{k}{\| \AA \|^2_F}p_i^T\MM p_i + \frac{\epsilon\| \AA \|^2_F}{k}\ell^T\ell + 4\epsilon
    \end{split}
\end{equation}
Recall $p_i$ lies in the column space of $\U_{\setminus \ell}$, and thus $p_i\MM p_i \leq \frac{\| A\|^2_F}{k}$ and thus we get 
\begin{equation}
\frac{2\sqrt{k}}{\sigma_i \| \AA \|^2_F}p_i^T \CC \CC^T u_i \leq 8\epsilon + \frac{\epsilon\| \AA \|^2_F}{k}\ell^T\ell \leq 12\epsilon
\end{equation}
Assuming again that $\epsilon < 1/2$ and observing that $\| \ell\|^2_2 \leq \frac{4k}{\epsilon\| \AA \|^2_F}$ we get 
\begin{equation}
    \begin{split}
        & \frac{2\sqrt{k}p_i\CC\CC^Tu_i}{\sigma_i\| \AA \|_F} \leq 12 \epsilon \\
        & (p_i\CC \CC^Tu_i^T)^2 \leq 144\epsilon^2 \frac{\sigma^2_i\| \AA \|^2_F}{k}
    \end{split}
\end{equation}
Plugging this back into $(\ref{eqn:cross4})$, we get that 
\begin{equation}
    \| \P_{\setminus \ell} \CC \CC^T \mathbf{U}_\ell \mathbf{\Sigma}^{-1}_\ell\|^2_F \leq 288\epsilon^2 \| \AA \|^2_F
\end{equation}
Since we have now bounded the second term, we plug it back into (\ref{eqn:cross2}),
\begin{equation}
\label{eqn:cross_final}
    \begin{split}
    \trace{\Y \LL \P_\ell \CC \CC^T \P_{\setminus \ell} \Y} & \leq \sqrt{\trace{\Y \MM \Y}}\sqrt{288\|\AA\|^2_F} 
    \leq 17\epsilon \trace{\Y\AA\AA^T\Y}
    \end{split}
\end{equation}
Combining $(\ref{eqn:split2})$ $(\ref{eqn:head_main})$, $(\ref{eqn:tail_final})$ and $(\ref{eqn:cross_final})$, we get 
\begin{equation}
\begin{split}
   \trace{\Y  \CC \CC^T  \Y} & =
   \trace{\Y \AA_{\ell} \AA_{\ell}^T\Y} +\trace{\Y \AA_{\setminus \ell} \AA_{\setminus \ell}^T\Y} \pm 40\epsilon \trace{\Y \AA \AA^T\Y} \pm 10\epsilon\| \AA \|^2_F
\end{split}
\end{equation}
Since $\trace{\Y \AA \AA^T\Y} \leq \| \AA \|^2_F$, rescaling $\epsilon$ by a constant finishes the proof.  
\end{proof}

We note that the critical ingredient in the proofs was estimating the column norms in sublinear time. We observe that we can also estimate the row norms in sublinear time and immediately obtain a Row Projection-Cost Preservation Theorem.

\begin{corollary}(Row Projection-Cost Preservation.)
\label{thm:pcp_main_row}
Let $\AA$ be a $m \times n$ matrix such that $\AA$ satisfies approximate triangle inequality. For $i \in [n]$, let $\widetilde{X}_i$ be an $O\left( \frac{n}{b} \right)$-approximate estimate for the $i^{th}$ row of $\AA$ such that it satisfies the guarantee of Lemma \ref{lem:row_norm_estimation}. Then, let $p = \{p_1, p_2 \ldots p_n\}$ be a probability distribution over the rows of $\AA$ such that $p_i = \frac{\widetilde{X}_i}{\sum_{i'}\widetilde{X}_{i'}}$. Let $t = O\left(\frac{ n k^2}{b\epsilon^2}\log(\frac{n}{\delta})\right)$. Then, construct $\CC$ using $t$ rows of $\AA$ and set each one to $\frac{A_{i,*}}{\sqrt{t p_i}}$ with probability $p_i$. With probability at least $1-\delta$, for any rank-$k$ orthogonal projection $\X$
\[
\| \CC - \CC\X \|^2_F = \| \AA - \AA\X \|^2_F \pm \epsilon\| \AA \|^2_F
\]
\end{corollary}
%O\left(\frac{ n k^2}{b\epsilon^2}\log(\frac{n}{\delta})\right)

Next, we describe how to apply projection-cost preserving sketching for low-rank approximation. Let $\CC$ be a \textit{column pcp} for $\AA$. Then, an approximate solution for the best rank-$k$ approximation to $\CC$ is an approximate solution for the best rank-$k$ approximation to $\AA$. Formally, 

\begin{lemma}
\label{lem:pcp_useful}
Let $\CC$ be a \textit{column pcp} for $\AA$ satisfying the guarantee of Theorem \ref{thm:pcp_main_col}. Let $\P^*_{\CC}$ be the projection matrix that minimizes $\| \CC - \X\CC \|^2_F$ and $\P^*_{\AA}$ be the projection matrix that minimizes $\|\AA - \X\AA \|^2_F$. Then, for any projection matrix $\P$ such that $\|\CC - \P \CC \|^2_F \leq  \| \CC - \P^*_{\CC} \CC\|^2_F + \epsilon \|\CC\|^2_F$, with probability at least $98/100$,
\[
\|\AA - \P \AA \|^2_F \leq  \| \AA - \P^*_{\AA} \AA\|^2_F + \epsilon \| \AA \|^2_F
\]
A similar guarantee holds if $\CC$ is a \textit{row pcp} of $\AA$. 
\end{lemma}
\begin{proof}
By the optimality of $\P^*_C$, we know that $\|\CC - \P \CC \|^2_F \leq  \| \CC - \P^*_{\CC} \CC\|^2_F + \epsilon\|\CC \|^2_F \leq  \| \CC - \P^*_{\AA} \CC\|^2_F + \epsilon \| \CC\|^2_F$. Since $\CC$ is a \textit{column pcp} of $\AA$, $\| \CC - \P^*_{\AA} \CC\|^2_F \leq \| \AA - \P^*_{\AA} \AA\|^2_F + \epsilon \|\AA\|^2_F$, therefore, with probability at least $99/100$, 
\begin{equation}
\|\CC - \P \CC \|^2_F \leq   \| \AA - \P^*_{\AA} \AA\|^2_F + \epsilon \|\AA\|^2_F +\epsilon \|\CC\|^2_F \leq \| \AA - \P^*_{\AA} \AA\|^2_F + O(\epsilon) \|\AA\|^2_F 
\end{equation}
where the last inequality follows from $\expec{}{\CC} = \AA$ and Markov's bound. 
Similarly, $\| \CC - \P \CC  \|^2_F \geq \| \AA - \P \AA  \|^2_F - \epsilon \|\AA \|^2_F$, therefore, with probability at least $99/100$,
\begin{equation}
\|\CC - \P \CC \|^2_F \geq \|\AA - \P \AA \|^2_F - O(\epsilon) \| \AA \|^2_F 
\end{equation}
Union bounding over the two events and combining the two equations, with probability at least $98/100$ we get 
\begin{equation}
\|\AA - \P \AA \|^2_F   \leq \| \AA - \P^*_{\AA} \AA\|^2_F + 4\epsilon \|\AA\|^2_F 
\end{equation}
Rescaling $\epsilon$ completes the proof. We note that a similar lemma holds if $\CC$  is a \textit{row pcp} of $\AA$.  
\end{proof}

%% file: full_algorithm.tex
\section{Optimizing the Exponent}

In this section we improve the running time of our previous sublinear time algorithm. Intuitively, our algorithm recursively constructs projection-cost preserving sketches for the rows and columns of the original matrix by sampling according to coarse estimates of the row and column norms. Note, we are able to obtain these estimates by dividing the subsampled matrices at each step into weight classes such that each weight class approximately satisfies triangle inequality. At the bottom of the recursion we reduce the input matrix $\AA$ to a $\textrm{poly}(\frac{k}{\epsilon}) \times \textrm{poly}(\frac{k}{\epsilon})$ matrix, for which we can compute the SVD in $O( \textrm{poly}(\frac{k}{\epsilon}) )$ time. 

Starting with an orthonormal basis of the \texttt{SVD}, we alternate between approximately computing the best rank-$k$ projection in the column and the row space all the way up the recursion chain and output the final rank-$k$ matrix. However, computing the \texttt{SVD} or even running an input-sparsity time algorithm near the top becomes prohibitively expensive and is no longer sublinear. 
Therefore, we find approximate solutions to the best rank-$k$ column and row subspaces by formulating a regression problem, sketching it to a smaller dimension using leverage score sampling and solving it approximately.

We show that recursive sampling indeed approximately preserves rank-$k$  subspaces of the row and column space. For the sake of brevity throughout the rest of the analysis, let $\AA_{(i)}$ be a $t_i \times s_i$ matrix created by recursively sampling rows or columns of $\AA_{(i-1)}$ such that at each step the \textit{row} or \textit{column pcp} properties are satisfied. Formally, let $\AA_{(0)} = \AA$ be a $t_0 \times s_0$ matrix, where $t_0 = m$ and $s_0 = n$. Then, if $i$ is odd, $\AA_{(i)} = \AA_{(i-1)} \S_i$ is a $t_{i-1} \times s_{i} $ matrix and a \textit{column pcp} for $\AA_{(i-1)}$ and if $i$ is even, $\AA_{(i)} = \T_i \AA_{(i-1)} $ is a $t_i \times s_{i-1}$ matrix and a \textit{row pcp} for $\AA_{(i-1)}$. We note that $s_i = \widetilde{\Theta}\left( \frac{s_{i-1}}{b_2}\textrm{poly}(\frac{k}{\epsilon})\right) $ and $t_i = \widetilde{\Theta}\left( \frac{t_{i-1}}{b_1}\textrm{poly}(\frac{k}{\epsilon}) \right)$.

To address the issue of rescaling every time we subsample rows or columns, we split the rows or columns of $\AA_{(i)}$ into $O(\epsilon^{-1} \log(mn))$ weight classes such that triangle inequality approximately holds in each weight class. Therefore, the column or row norm estimation algorithm goes through for each weight class independently and we obtain an overall $O(\frac{m}{b_2})$-approximation to the column norms at the cost of reading $O(\epsilon^{-1} \log(m) b_2)$ entries per column. A similar guarantee holds for the row norms. This idea simply extends to the recursive algorithm as we can create weight classes at each recursive step run the simple row and column norm estimation algorithms. 

\begin{Frame}[\textbf{Algorithm \ref{alg:final_sublinear} : Full Sublinear Time Algorithm.}]
\label{alg:final_sublinear} 
\textbf{Input:} A Distance Matrix $\AA_{m \times n}$, integer $k$,  $\epsilon > 0$ and a small constant $\gamma >0$. 

\begin{enumerate}
	\item  Let $2r =O\left(1/\gamma\right)$ and let $\AA_{(0)} = \AA$, $s_0 = n$ and $t_0 = m$. Set $b_1 = m^{\gamma}$, $b_2 = n^\gamma$.
	\item For $i \in [2r]$, recursively construct matrix $\AA_{(i)}$ as follows:
    \begin{enumerate}
    	\item If $i$ is odd, run \texttt{ColumnNormEstimation}$(\AA_{(i-1)}, b_2)$. By Lemma \ref{lem:rescaling}, we obtain $O\left(\frac{t_{i-1}}{b_2}\right)$ -approximate estimates of the column norms of $\AA_{(i-1)}$. Let $q = \{ q_1, q_2 \ldots q_{s_{i-1}} \} $ denote a distribution over columns of $\AA_{(i-1)}$ proportional to the relative estimate for each column. 
        
        Construct a \textit{column pcp} for $\AA_{(i-1)}$ by sampling $s_i =  \widetilde{\Theta}\left( \frac{s_{i-1}}{b_2}\textrm{poly}(\frac{k}{\epsilon})\right)$ columns such that each column is set to $\frac{(\AA_{(i-1)})_{*,j}}{\sqrt{s_i q_j}}$ with probability $q_j$. Let $\AA_{(i)} = \AA_{(i-1)}\S_i$ be the resulting $t_{i-1} \times s_i$ matrix that follows guarantees of Theorem \ref{thm:pcp_main_col}. 
        \item If $i$ is even, run \texttt{RowNormEstimation}$(\AA_{(i-1)}, b_1)$. By Lemma \ref{lem:rescaling}, we obtain $O\left(\frac{s_{-1}}{b_1}\right)$ -approximate estimates of the row norms of $\AA_{(i-1)}$. Let $p = \{ p_1, p_2 \ldots p_{t_{i-1}} \} $ denote a distribution over rows of $\AA_{(i-1)}$ proportional to the relative estimate for each row. 
  
  Construct a \textit{row pcp} for $\AA_{(i-1)}$ by sampling $t_i = \widetilde{\Theta}\left( \frac{t_{i-1}}{b_1}\textrm{poly}(\frac{k}{\epsilon}) \right)$ columns such that each column is set to $\frac{(\AA_{(i-1)})_{\ell,*}}{\sqrt{t_i p_{\ell}}}$ with probability $p_{\ell}$. Let $\AA_{(i)} = \T_i\AA_{(i-1)}$ be the resulting $t_i \times s_{i-1}$ matrix that follows guarantees of Corollary \ref{thm:pcp_main_row}.
    \end{enumerate}
    \item  Let $\AA_{(2r)}$ be the final matrix at the end of the recursion. Compute the truncated \texttt{SVD} $(\AA_{(2r)}, k) = \U_{2r} \Sigma_{2r} \V^T_{2r}$. Let $\V^T_{2r}$ represent the top $k$ singular vectors in the row space. Construct a leverage-score sketching matrix $\E_{2r}$ with  poly$(\frac{k}{\epsilon})$ columns using the leverage scores of $\V^T_{2r}$, following the guarantees of Theorem \ref{thm:clarkson_woodruff_regression}. Compute $\X_{\AA_{(2r-1)}} = \argmin_{\X}\| \AA_{(2r-1)}\E_{2r-1} - \X \V^T_{2r}\E_{2r-1}\|$. 
	\item Compute a decomposition $\U_{2r-1} \V_{2r-1}$ of $\X_{\AA_{(2r-1)}}\V^T_{2r}$ such that $\U_{2r-1}$ has orthonormal columns. Consider the regression problem $\min_{\X} \|\AA_{(2r-2)} -  \U_{2r-1}\X \|^2_F$. Compute $\X_{\AA_{(2r-2)}}$ $= \argmin_{\X} \|\E_{2r-2}\AA_{(2r-2)} -  \E_{2r-2}\U_{2r-1}\X \|^2_F $, where $\E_{2r-2}$ is a leverage score sketching matrix with poly$\left(\frac{k}{\epsilon}\right)$ rows constructed according to Theorem \ref{thm:clarkson_woodruff_regression}.
    \item Let $\U_{2r-1}\X_{\AA_{(2r-2)}}$ be the starting point for $\AA_{2r-3}$ as $\U_{2r} \Sigma_{2r} \V^T_{2r}$ was for $\AA_{2r-1}$ and recurse to the top. Let $\U_{1}$, $\X_{\AA}$ be the solution obtain from solving $\min_{\X} \|\E\AA -  \E\U_{1}\X \|^2_F $, where $\E$ is a leverage score sketching matrix with poly$\left(\frac{k}{\epsilon}\right)$ rows, constructed according to Theorem \ref{thm:clarkson_woodruff_regression}.
\end{enumerate}
\textbf{Output:} $\MM= \U_{1}$, $\NN^T = \X_{\AA}$
\end{Frame}

Intuitively, to handle the first column rescaling , we partition the columns of the matrix into $O(\epsilon^{-1} \log(mn))$ blocks such that each block satisfies approximate triangle inequality. By Lemma \ref{lem:row_norm_estimation}, we can estimate the row norms for each block efficiently, and summing the estimates suffices to obtain an approximation to the row norms. Next, we subsample the rows and scale them. However, we observe that we can yet again partition each sub-matrix that satisfies approximate triangle inequality into $O(\epsilon^{-1} \log(mn))$ weight classes and yet again satisfy approximate triangle inequality. We note that we only recurse a constant $(1/\gamma)$ number of times, therefore the total number of sub-matrices formed is $O\left((\epsilon^{-1} \log(mn))^\frac{1}{\gamma}\right) = \textrm{poly}(\epsilon^{-1} \log(mn))$ and the run-time blows up by at most that factor.   

%Therefore, the total number of entries read is $O(\epsilon^{-1}\log(mn) b_2 n)$ and $O(\epsilon^{-1}\log(mn) b_1 m )$ to estimate column norms and row norms respectively. 

\begin{lemma}(Estimating row and column norms under rescaling.)
\label{lem:rescaling}
Let $\AA$ be a $m \times n$ matrix such that $\AA$ satisfies approximate triangle inequality. 
For a small fixed constant $\gamma$, and for all $i \in [1/\gamma]$, let $\AA_{(i)}$ be a $t_i \times s_i$ scaled sub-matrix of $\AA$ as defined as above. There exists an algorithm that, with probability at least $9/10$, obtains a $O\left(\frac{s_{i-1}}{b_1}\right)$-approximation to the row norms of $\AA_{(i)}$ in $\widetilde{O}(b_2m + n)$ time. A similar guarantee holds for estimating column norms.
\end{lemma}
\begin{proof}
We begin with a $m \times n$ matrix $\AA$ that satisfies approximate triangle inequality. It follows from Corollary \ref{cor:column_norm_estimation} that we can approximately estimate column norms of $\AA$ in $O(b_2m + n)$ time. Therefore, we can construct the \textit{column pcp}, $\AA_{(1)}=\AA \S_1$, such that the $j^{th}$ column of $\AA$, if sampled, is rescaled by $\frac{1}{\sqrt{s_1 q_j}}$. Therefore, the resulting matrix, $\AA_{(1)}$ may no longer be a distance matrix. As discussed in the previous section, to address this issue, we partition the columns of $\AA_{(1)}$ into weight classes such that the $g^{th}$ weight class contains column index $j$ if the corresponding scaling factor $\frac{1}{\sqrt{ q_j}}$ lies in the interval  $\left[(1+\epsilon)^g, (1+\epsilon)^{g+1}\right)$. Note, we can ignore the $(\frac{1}{\sqrt{s_1}})$-factor since every entry is rescaled by the same constant. Formally, 
\begin{equation}
\mathcal{W}^{(1)}_g = \left\{ i \in [s_1] \textrm{ $\Big{|}$ } \frac{1}{\sqrt{q_j}} \in \left[(1+\epsilon)^g, (1+\epsilon)^{g+1}\right) \right\}
\end{equation}
Next, with high probability, for all $j\in[n]$, if column $j$ is sampled, $\frac{1}{q_j} \leq n^c$ for a large constant $c$. If instead, $q_j \leq \frac{1}{n^{c'}}$, the probability that the $j^{th}$ is sampled would be at most $1/n^{c'}$, for some $c' > c$. Union bounding over such events for $n$ columns, the number of weight classes is at most $\log_{1+\epsilon}(n^c) = $     $O\left(\epsilon^{-1}\log(n)\right)$.
Let $\AA_{(1)|\mathcal{W}^{(1)}_g}$ denote the columns of $\AA_{(1)}$ restricted to the set of indices in $\mathcal{W}_g$. Observe that all entries in $\AA_{(1)|\mathcal{W}^{(1)}_g}$ are scaled to within a $(1+\epsilon)$-factor of each other and therefore, satisfy approximate triangle inequality (equation \ref{eqn:approx_tri}).
Therefore, row norms of $\AA_{(1)|\mathcal{W}^{(1)}_g}$ can be computed using Algorithm \ref{alg:row_norm_estimation} and the estimator is an $O\left(\frac{n}{b_2}\right)$-approximation (for some parameter $b_2$), since Lemma \ref{lem:row_norm_estimation} blows up by a factor of at most $1+\epsilon$. Summing over the estimates from each partition above, with probability at least $99/100$, we obtain an $O\left(\frac{n}{b_2}\right)$-approximate estimate to row norms of $\AA_{(1)}$. However, we note that each iteration of Algorithm \ref{alg:row_norm_estimation} reads $b_2m + n$ entries of $\AA$ and there are at most $O(\epsilon^{-1}\log(n))$ iterations. Therefore, the time taken to compute the estimates to the row norms is $O\left( (b_2m + n)\epsilon^{-1}\log(n) \right)$. 

Now, we can construct a \emph{row pcp} of $\AA_{(1)}$, which we denote by $\AA_{(2)}$, such that each row in $\AA_{(2)}$ is a scaled subset of the rows of $\AA_{(1)}$. Our next task is to estimate the column norms of $\AA_{(2)}$. It suffices to show that $\AA_{(2)}$ can be partitioned into a small number of sub-matrices such that each one satisfies $(1+\epsilon)$-approximate triangle inequality. 

Observe, we previously split the matrix $\AA_{(1)}$ according to $\mathcal{W}^{(1)}_g$, into  $O\left(\epsilon^{-1}\log(n)\right)$ sub-matrices such that each matrix satisfies $(1+\epsilon)$-approximate triangle inequality. Consider one such sub-matrix, $\AA_{(1)|\mathcal{W}^{(1)}_g}$. 
In the construction of the \emph{row pcp} $\AA_{(2)}$, we rescale a subset of rows of each of $\AA_{(1)|\mathcal{W}^{(1)}_g}$. Therefore, we can again create $O\left(\epsilon^{-1}\log(m)\right)$ geometrically increasing weight classes for the rows of $\AA_{(1)|\mathcal{W}^{(1)}_g}$. Note, restricting rows of $\AA_{(1)|\mathcal{W}^{(1)}_g}$ to one weight class results in a sub-matrix that satisfies $(1+\epsilon)^2$-approximate triangle inequality. We repeat the above analysis for each such sub-matrix, since we again start with a matrix that satisfies $(1+\epsilon)$-approximate triangle inequality, after rescaling $\epsilon$ by a constant.  

Critically, we note that we only repeat the recursion a constant number of times, therefore the approximation factor for triangle inequality blows up by $(1+\epsilon)^{1/\gamma}$. Since $\gamma$ is a constant, we can rescale $\epsilon$ by a constant, and therefore, all sub-matrices satisfy $(1+\epsilon)$-approximate triangle inequality. the total number of sub-matrices formed is $O\left((\epsilon^{-1} \log(n)\log(m))^\frac{1}{\gamma}\right) = \textrm{poly}(\epsilon^{-1} \log(n)\log(m))$ and the run-time blows up by at most that factor. 
\end{proof}

Note, we can now black-box the algorithm for estimating row and column norms of scaled sub-matrices of $\AA$. For the sake of brevity, we do not include them in Algorithm \ref{alg:final_sublinear}. Next, we present a critical structural result that enables us to recursively apply the \textit{pcp} guarantees. 

\begin{lemma}(Recursive PCP Lemma.)
\label{lem:recursive_pcp}
Let $\AA_{(i)}$ be defined as above. Then, for any $\epsilon > 0$, integer $k$ and a small constant  $\gamma > 0$, $2r = O\left( \frac{1}{\gamma}\right)$ , simultaneously for all odd $i \in [2r]$, for all rank-$k$ projection matrices $\X_i$, with probability at least 98/100,
\[\| \AA_{(i)} (\II - \X_{i}) \|^2_F = \| \AA_{(i-1)} (\II - \X_{i}) \|^2_F \pm \epsilon\| \AA_{(i-1)} \|^2_F \] 
Further, let $\X^*_{\AA_{(i)}}$ be the projection that minimizes $\| \AA_{(i)} (\II - \X_{i}) \|^2_F$ and let $\X^*_{\AA_{(i-1)}}$ be the projection that minimizes $\|\AA_{(i-1)} (\II - \X_{i}) \|^2_F$. Then, simultaneously for all odd $i$, for any rank-$k$ projection matrix $\X_i$ such that $\| \AA_{(i)} (\II - \X_{i}) \|^2_F \leq \| \AA_{(i)} (\II - \X^*_{\AA_{(i)}}) \|^2_F + \epsilon \|\AA_{(i)}\|^2_F$, with probability at least $98/100$, 
\[
\| \AA_{(i-1)} (\II - \X_{i}) \|^2_F \leq \| \AA_{(i-1)} (\II - \X^*_{\AA_{(i-1)}}) \|^2_F + \epsilon \|\AA_{(i-1)}\|^2_F
\]
A similar guarantee holds if $i$ is even.
\end{lemma}

\begin{proof}
For a given matrix $\AA_{(i)}$, such that $i$ is odd, we employ Theorem \ref{thm:pcp_main_col} with $\delta = 1/n^c$ for a fixed constant $c$. Note, the number of columns we sample only blows up by a constant and 
\begin{equation*}
\| \AA_{(i)} (\II - \X_{i}) \|^2_F \leq \| \AA_{(i-1)} (\II - \X_{i}) \|^2_F + \epsilon\| \AA_{(i-1)} \|^2_F
\end{equation*}
holds with probability at least $1 - \frac{1}{n^c}$. Union bounding over all such events for $i \in [2r]$, such that $i$ is odd, the above guarantee holds simultaneously with probability at least $1 - 1/n^{c-1}$. Setting $\delta = 1/n^c$ in Corollary \ref{thm:pcp_main_row}, it follows that simultaneously for all $i\in [2r]$, such that $i$ is even, 
\begin{equation*}
\| (\II - \X_{i}) \AA_{(i)} \|^2_F \leq \|  (\II - \X_{i}) \AA_{(i-1)} \|^2_F + \epsilon\| \AA_{(i-1)} \|^2_F
\end{equation*}
with probability at least $1 - 1/n^{c-1}$.
Note, for a fixed odd $i$, following the analysis of Lemma \ref{lem:pcp_useful}, 
\begin{equation*}
\| \AA_{(i-1)} (\II - \X_{i}) \|^2_F \leq \| \AA_{(i-1)} (\II - \X^*_{\AA_{(i-1)}}) \|^2_F + O(\epsilon) \|\AA_{(i-1)}\|_F^2
\end{equation*}
holds with probability at least $1 - \frac{1}{200r}$. Union bounding over all such events $i \in [2r]$, such that i is odd, the above guarantee holds simultaneously, with probability at least $99/100$. 

Similarly, for even $i$, let $\X^*_{\AA_{(i)}}$ be the projection that minimizes $\|(\II - \X_{i}) \AA_{(i)}\|^2_F$ and let $\X^*_{\AA_{(i-1)}}$ be the projection that minimizes $\| (\II - \X_{i}) \AA_{(i-1)} \|^2_F$. Then, simultaneously for all even $i$, for any rank-$k$ projection matrix $\X_i$ such that 
\begin{equation*}
\| (\II - \X_{i})\AA_{(i)}  \|^2_F \leq \| (\II - \X^*_{\AA_{(i)}}) \AA_{(i)}  \|^2_F + O(\epsilon)\|\AA_{(i)}\|_F^2
\end{equation*}
with probability at least $99/100$, it holds that 
\begin{equation*}
\| (\II - \X_{i}) \AA_{(i-1)}  \|^2_F \leq \| (\II - \X^*_{\AA_{(i-1)}}) \AA_{(i-1)} \|^2_F + O(\epsilon) \|\AA_{(i-1)}\|_F^2
\end{equation*}
Union bounding over the odd and even events completes the proof.
\end{proof}

Using the above structural guarantees, we show that Algorithm \ref{alg:final_sublinear} indeed achieves an additive error guarantee for low-rank approximation. Intuitively, at each recursive step we either approximately preserve all rank-$k$ row or column projections. We finally obtain a matrix that is independent of $m$ and $n$ and therefore we can compute its SVD. We then begin with the rank-$k$ matrix output by the SVD and critically rely on Lemma \ref{lem:pcp_useful} to switch between finding approximately optimal projections in the row and column space while climb back up the recursive stack.  
  
\begin{lemma}(Additive Error Guarantee.)
\label{lem:guarantee}
Let $\AA \in \mathbb{R}^{m \times n}$ be a matrix such that it satisfies approximate triangle inequality. Then, for any $\epsilon > 0$, integer $k$, and a small constant $ \gamma >0$, with probability at least 9/10, Algorithm $\ref{alg:final_sublinear}$ outputs a rank-$k$ matrix $\MM\NN^T$ such that $\MM \in \mathbf{R}^{m\times k}$, $\NN \in \mathbf{R}^{n\times k}$ and
\[\| \AA - \MM \NN^T \|^2_{F} \leq  \| \AA - \AA_k\|^2_{F} + O(\epsilon) \| \AA \|^2_F\]
\end{lemma}
\begin{proof}
By Lemma \ref{lem:recursive_pcp}, we know that approximately optimal rank-$k$ projection matrix for $\AA_{(i)}$ is an approximately optimal rank-$k$ matrix for $\AA_{(i-1)}$ up to an additive error term of $\epsilon\| \AA_{(i-1)}\|^2_F$. Let $2r = O\left(\frac{1}{\gamma}\right)$ be the number of recursive calls. Concretely, since $\AA_{(2r)}$ is a \textit{row pcp} of $\AA_{(2r-1)}$, we know that for all rank-$k$ projections $\X_{2r}$, with probability at least $98/100$, 
\begin{equation}\label{eqn:level2r}
    \| \AA_{(2r)} (\II - \X_{2r}) \|^2_F \leq \| \AA_{(2r-1)} (\II - \P_{2r}) \|^2_F + \epsilon\| \AA_{(2r-1)} \|^2_F
\end{equation}
Since $\AA_{(2r)}$ is a small matrix of dimension, we can afford to compute $\texttt{SVD}(\AA_{(2r)})$ (we analyze this runtime below). Let $\U_{2r} \D_{2r} \V^T_{2r}$ be the truncated \texttt{SVD}, containing the top $k$ singular values and setting the rest to $0$. Thus, we know that $\V_{2r} \V^T_{2r}$ is the optimal projection matrix for $\AA_{(2r)}$, and by Lemma \ref{lem:recursive_pcp},
\begin{equation}
\label{eqn:level2r_2}
    \| \AA_{(2r-1)}(\II - \V_{2r} \V^T_{2r} ) \|^2_F \leq \| \AA_{(2r-1)} - (\AA_{(2r-1)})_k \|^2_F + \epsilon\| \AA_{(2r-1)} \|^2_F
\end{equation}

As discussed in the analysis of Algorithm \ref{alg:first_sublinear}, observe that computing the \texttt{SVD} or even running an input-sparsity time algorithm as we recurse back up to the top becomes prohibitively expensive and no longer sublinear. 
Therefore, we follow the previous strategy of setting up a regression problem, sketching it and solving it approximately using leverage scores.
We observe that an approximately optimal solution for $\AA_{(2r-1)}$ lies in the row space of $\V^T_{2r}$ and set up the following regression problem: 
\begin{equation}
\min_{\X} \|\AA_{(2r-1)} - \X \V^T_{2r} \|^2_F
\end{equation}
Though this problem is small and independent of $m$ and $n$, as we recurse up, the regression problems grow larger and larger. Therefore, we sketch it using the leverage scores of $\V^T_{2r}$. Note, since $\V^T_{2r}$ is orthonormal, the leverage scores are precomputed. With probability at least $98/100$, we compute $\X_{\AA_{(2r-1)}} = \argmin_{\X}\| \AA_{(2r-1)}\E_{2r-1} - \X \V^T_{2r}\E_{2r-1}\|$, where $\E_{2r-1}$ is a leverage score sketching matrix with poly$\left(\frac{k}{\epsilon}\right)$ columns. Given the sketching guarantee of Theorem \ref{thm:clarkson_woodruff_regression}, 
\begin{equation}
\label{eqn:bottom_level}
\begin{split}
\|\AA_{(2r-1)} - \X_{\AA_{(2r-1)}}\V^T_{2r} \|^2_F & \leq (1+\epsilon) \min_{\X} \| \AA_{(2r-1)} - \X\V^T_{2r} \|^2_F \\
& \leq (1+\epsilon) \| \AA_{(2r-1)} -  \AA_{(2r-1)} \V_{2r}\V^T_{2r} \|^2_F \\
& \leq \| \AA_{(2r-1)} - (\AA_{(2r-1)})_k \|^2_F + O(\epsilon)\|\AA_{(2r-1)}\|^2_F
\end{split}
\end{equation}
where the last inequality follows from equation \ref{eqn:level2r_2}.Therefore, $\X_{\AA_{(2r-1)}}\V^T_{2r}$ has rank at most $k$ and is an approximate optimal solution for $\AA_{(2r-1)}$. Applying Lemma \ref{lem:recursive_pcp} again, we know that projecting onto the column space of $\X_{\AA_{(2r-1)}}\V^T_{2r}$ is an approximately optimal solution for $\AA_{(2r-2)}$. Formally, let $\X_{\AA_{(2r-1)}}\V^T_{2r} = \U_{2r-1} \V_{2r-1}$ such that $\U_{2r-1}$ has orthonormal columns. It follows from equation \ref{eqn:bottom_level} that $\|\AA_{(2r-1)} - \U_{2r-1} \V_{2r-1} \|^2_F =  \| \AA_{(2r-1)} - (\AA_{(2r-1)})_k \|^2_F  \pm \epsilon \|\AA_{(2r-1)}\|^2_F$. Therefore, 
\begin{equation}
 \| (\II - \U_{2r-1} \U^T_{2r-1} ) \AA_{(2r-2)} \|^2_F \leq \| \AA_{(2r-2)} - (\AA_{(2r-2)})_k \|^2_F + \epsilon\| \AA_{(2r-2)} \|^2_F
\end{equation}
We observe that a good solution for $\AA_{(2r-2)}$ exists in the column space of $\U_{2r-1}$, and set up the following regression problem:
\begin{equation}
\min_{\X} \|\AA_{(2r-2)} -  \U_{2r-1}\X \|^2_F
\end{equation}
Again, we sketch this regression problem using the leverage scores of $\U_{2r-1}$. Recall, $\U_{2r-1}$ has orthonormal rows and thus the leverage scores are precomputed. We can then compute $\X_{\AA_{(2r-2)}} = \argmin_{\X} \|\E_{2r-2}\AA_{(2r-2)} -  \E_{2r-2}\U_{2r-1}\X \|^2_F $, where $\E_{2r-2}$ is a leverage score sketching matrix with poly$\left(\frac{k}{\epsilon}\right)$ rows. Given the sketching guarantee of leverage score sampling in Theorem \ref{thm:clarkson_woodruff_regression},
\begin{equation}
\label{eqn:2r_1_level}
\begin{split}
\|\AA_{(2r-2)} - \U_{2r-1}\X_{\AA_{(2r-2)}} \|^2_F & \leq (1+\epsilon) \min_{\X} \| \AA_{(2r-2)} - \U_{2r-1}\X\|^2_F \\
& \leq (1+\epsilon) \| \AA_{(2r-2)} -  \U_{2r-1} \U_{2r-1}^T  \AA_{(2r-2)}  \|^2_F \\
& \leq \| \AA_{(2r-2)} - (\AA_{(2r-2)})_k \|^2_F + O(\epsilon)\|\AA_{(2r-2)}\|^2_F
\end{split}
\end{equation}
We observe that $\U_{2r-1}\X_{\AA_{(2r-2)}}$ is a rank-$k$ matrix, written in factored from, that is an approximate solution for $\AA_{(2r-2)}$. Using $\U_{2r-1}$, $\X_{\AA_{(2r-2)}}$, we can repeat the above analysis $r$ times all the way up the recursion stack. Note, the last level of the analysis is simply the one presented in the proof of Theorem \ref{thm:first_sublinear}. Let $\U_{1}$, $\X_{\AA}$ be the solution obtain from solving the final regression problem. Note, union bounding over all the random events above, we obtain $\U_{1}$, $\X_{\AA}$ with probability at least $9/10$. Setting $\MM = \U_{1}$ and $\NN^T = \X_{\AA}$ finishes the proof and satisfied the \textit{additive error }guarantee for $\AA$.
\end{proof}
Next, we show the running time of Algorithm \ref{alg:final_sublinear} is sublinear in $m$ and $n$ for an appropriate setting of $b_1$, $b_2$ and $\gamma$. 
\begin{lemma}
\label{lem:runtime}
Let $\AA \in \mathbb{R}^{m \times n}$ be a matrix such that it satisfies approximate triangle inequality. Then, for any $\epsilon > 0$, integer $k$, and a small constant $\gamma > 0$, Algorithm $\ref{alg:final_sublinear}$ runs in  $\widetilde{O}\left( \left(m^{1+\gamma} + n^{1+\gamma}\right)\textrm{poly}(\frac{k}{\epsilon})\right)$ time. 
\end{lemma}
\begin{proof}
Recall, from the running time analysis of Algorithm \ref{alg:first_sublinear}, for a $t_i \times s_i$ matrix $\AA_{(i)}$, we a construct \textit{column \textrm{and} row pcp} in $O\left( \frac{t_i\log(t_i)}{b_1}\text{poly}(\frac{k}{\epsilon}) + b_1s_i + t_i \right)$ and $O\left(\frac{s_i \log(s_i)}{b_2}\textrm{poly}(\frac{k}{\epsilon}) + (b_2t_i + s_i )\right)$ respectively. Since $i \in [2r]$ and $2r = O(1/\gamma)$, we can compute $\AA_{(2r)}$ in $O\left(\left( b_1n + b_2m\right)\textrm{poly}(\frac{k\log(mn)}{\epsilon})\right)$, since the running time is dominated by sampling $b_1$ entries in each column and $b_2$ entries in each row of the input matrix $\AA$, at the top level. Note, $\AA_{(2r)}$ is a $t_{2r} \times s_{2r}$ matrix, where $t_{2r} = \frac{m}{b_2^{1/\gamma}}\textrm{poly}(\frac{k\log(m)}{\epsilon})$ and $s_{2r}=\frac{n}{b_1^{1/\gamma}}\textrm{poly}(\frac{k\log(n)}{\epsilon})$.
We can compute the SVD of $\AA_{(2r)}$ in $O\left( \left(\frac{nm}{(b_1 b_2)^{1/\gamma}}\right)^2 \textrm{poly}(\frac{k\log(m)\log(m)}{\epsilon}) \right)$. Next, we solve $2r$ regression problems by sketching them as we recurse back up. We again upper bound each recursive step by the running time of the top level. Recall, from the analysis of Algorithm \ref{alg:first_sublinear}, the regression problem can be solved in $O\left( (m + n)\textrm{poly}(\frac{k\log(m)\log(n)}{\epsilon}) \right)$. Setting $b_1 = n^{\gamma}$ and $b_2 = m^{\gamma}$, the overall running time is $\widetilde{O}\left(\left(m^{1+\gamma} + n^{1+\gamma}\right)\textrm{poly}(\frac{k}{\epsilon})\right)$.   
\end{proof}
Lemma $\ref{lem:guarantee}$ together with Lemma $\ref{lem:runtime}$ imply our main theorem:

\begin{theorem}(Sublinear Low-Rank Approximation for Distance Matrices.)
\label{thm:sublinear_lra}
Let $\AA \in \mathbb{R}^{m \times n}$ be a matrix such that it satisfies approximate triangle inequality. Then, for any $\epsilon > 0$, integer $k$ and a small constant $\gamma > 0$, there exists an algorithm that accesses $O\left(m^{1+\gamma}+ n^{1+\gamma}\right)$ entries of $\AA$ and runs in time $\widetilde{O}\left( \left(m^{1+\gamma} + n^{1+\gamma}\right)\textrm{poly}(\frac{k}{\epsilon})\right)$ to output matrices $\MM \in \mathbb{R}^{m \times k}$ and $\NN \in \mathbb{R}^{n \times k}$ such that with probability at least $9/10$,
\[\| \AA - \MM \NN^T \|^2_{F} \leq  \| \AA - \AA_k\|^2_{F} + \epsilon \| \AA \|^2_F\]
\end{theorem}

%% file: relative_error.tex
\section{Relative Error Guarantees}
In this section, we consider the \emph{relative error} guarantee \ref{eqn:relative} for distance matrices. We begin by showing a lower bound for any relative error approximation for distance matrices. We also preclude the possibility of a sublinear bi-criteria algorithm outputting a rank-poly$(k)$ matrix satisfying the rank-$k$ relative error guarantee. 

\begin{theorem}(Lower bound.)
Let $\AA$ be an $n \times n$ distance matrix. Let $\BB$ be a rank-$\textrm{poly}(k)$ matrix such that $\| \AA - \BB\|^2_F \leq c \| \AA - \AA_k \|$ for any constant $c > 1$. Then, any algorithm that outputs such a $\BB$ requires $\Omega(nnz(\AA))$ time.  
\end{theorem}
\begin{proof}
Let $P = \{ e_1, e_2 \ldots e_n \}$ be a set of $n$ standard unit vectors and $Q$ be a set of $n-1$ zero vectors along with one point $q$ such that $q = \pm e_i$, where $i$ and the sign are chosen uniformly at random. Let the underlying metric space be $\ell_\infty$-norm. Note, all but one pairwise distances in $\AA$ are $1$. Further, one entry in the $i^{th}$ row of $\AA$ is either $2$ or $0$. Note, $\AA$ is a rank-$2$ matrix and thus $\| \AA - \AA_k \|^2_F = 0$, for all $k >1$. Let $\BB$ be a rank-poly($k$) matrix that obtains any relative-error guarantee. Then, $\BB$ must exactly recover $\AA$ and therefore any algorithm needs to read all entries of $\AA$ in order to find the entry that is $0$ or $2$.    
\end{proof}

\subsection{Euclidean Distance Matrices}
We show that in the special case of Euclidean distances, when the entries correspond to {\it squared} distances, there exists a bi-criteria algorithm that outputs a rank-$(k+4)$ matrix satisfying the relative error rank-$k$ low-rank approximation guarantee. Note, here the point sets $P$ and $Q$ are identical. 

\begin{Frame}[\textbf{Algorithm \ref{alg:bi_criteria} : Bi-criteria Algorithm for Euclidean Matrices.}]
\label{alg:bi_criteria}
\textbf{Input:} A Euclidean Distance Matrix $\AA_{n \times n}$, integer $k$ and $\epsilon >0$.
\begin{enumerate}
	\item Let $\AA = \AA_1 + \AA_2 - 2\BB$ s.t. $\AA_1$ and $\AA_2$ are rank-$1$ matrices and $\BB$ is  a PSD Matrix. 
	\item Then, $\AA_1 = \mathbf{a}'_1 \mathbf{a}_1^T$ and $\AA_2 = \mathbf{a}'_2 \mathbf{a}_2^T$.
    \item Compute $\MM \NN^T $ by running the sublinear low-rank approximation algorithm from \cite{mw17} on $\BB$ with parameter $k+2$.
    \item Compute $\mathbf{V}$ an orthonormal basis for $\MM\NN^T$.  
	\item Let $\WW $ be $\mathbf{V}$ concatenated with $\mathbf{a}_1$ and $\mathbf{a}_2$. We denote $\WW $ as $[\mathbf{V}; \mathbf{a}_1, \mathbf{a}_2]$.
\end{enumerate}
\textbf{Output:} $\WW\WW^T$
\end{Frame}

Let $\AA$ be such a matrix for the point set $P$ s.t. $\AA_{i,j} = \| x_i - x_j\|^2_2 = \| x_i\|^2_2 + \| x_j\|^2_2 - 2 \langle x_i, x_j \rangle$. Then, we can write $\AA$ as a sum of three matrices $\AA_1$, $\AA_2$ and $\BB$ such that each entry in the $i^{th}$ row of $\AA_1$ is $\| x_i\|^2_2$, each entry in the $j^{th}$ column of $\AA_2$ is $\| x_j\|^2_2$ and $\BB$ is a Positive Semi-Definite (PSD) matrix, where $\BB_{i,j} = \langle x_i, x_j \rangle$. Therefore, we can represent $\AA$ as $\AA_1 + \AA_2 - 2\BB$. The main ingredient we use is the sublinear low-rank approximation of PSD matrices from \cite{mw17}.

\begin{theorem}(Musco-Woodruff Low-Rank Approximation.)
\label{thm:musco_woodruff}
Given a PSD matrix $\AA \in R^{n\times n}$, any $\epsilon<1$ and an integer $k$, there is an algorithm that runs in $O(n\textrm{poly}(\frac{k}{\epsilon}))$ and with probability at least $9/10$, outputs matrices $\MM, \NN \in \mathbb{R}^{n \times k}$ such that
\[\| \AA - \MM\NN^T \|^2_{F} \leq (1 +\epsilon) \| \AA - \AA_k\|^2_{F}\]
\end{theorem}

We now show that there exists an algorithm that outputs the description of a rank-$(k+4)$ matrix $\AA \WW\WW^T$ in sublinear time such it satisfies the relative-error rank-$k$ low rank approximation guarantee.

\begin{lemma}
Let $\AA$ be a Euclidean Distance matrix. Then, for any $\epsilon>0$ and integer $k$, with probability at least $9/10$, Algorithm $\ref{alg:bi_criteria}$ outputs a rank $(k+4)$ matrix $ \WW\WW^T$ such that \[\| \AA-\AA \WW\WW^T\|_F \leq (1+\epsilon)\| \AA - \AA_k\|_F\]
where $\AA_k$ is the best rank-k approximation to $\AA$. Further, Algorithm $\ref{alg:bi_criteria}$ runs in $O(n\textrm{poly}(\frac{k}{\epsilon}))$. 
\end{lemma}
\begin{proof}
First, we show that we can simulate the sublinear PSD algorithm from \cite{mw17} on $\BB$, given random access to $\AA$ and reading $n-1$ additional entries. Observe that computing $\BB_{i,j}$ requires $\| x_i\|^2_2$ and $\| x_j\|^2_2$. Since pairwise distances are invariant to a uniform shift in position. Therefore, w.l.o.g. we can assume that the first point, $x_1$, is at the origin. Then, the $j^{th}$ entry of the first row $\AA$ is $\| x_j\|^2_2$. Now, $\BB_{i,j} = (\| x_i\|^2_2 + \| x_j\|^2_2 - \AA_{i,j})/2$ and we have access to each of these values. Therefore, we can simulate the algorithm for sublinear low-rank approximation on $\BB$ in $O(n\textrm{poly}(\frac{k}{\epsilon}))$. 
Note, in Algorithm \ref{alg:bi_criteria} we find a rank $k+2$ approximation to $\BB$ and by Theorem $\ref{thm:musco_woodruff}$, the algorithm outputs matrices $\MM, \NN \in \mathbb{R}^{n \times k}$ such that 
\[
\| \BB(\II - \MM \NN^T)\|_F \leq (1+\epsilon)\| \BB - \BB_{k+2}\|_F
\]
where $\BB_{k+2}$ is the best rank-$(k+2)$ approximation to $\BB$. Let $\V$ be an orthonormal basis for $\MM$, which can be computed in $O(n k^2)$ time. Let $\WW =[\mathbf{V}; \mathbf{a}_1, \mathbf{a}_2]$. We observe that the projection matrix $\WW \WW^T$ applied to $\AA_1$ and $\AA_2$ yields $\AA_1$ and $\AA_2$ respectively since they lie in the column space of $\WW$. Therefore, 
\begin{equation}
\begin{split}
\| \AA(\II - \WW \WW^T)\|_F & = \| (\AA_1 + \AA_2 - 2\BB)(\II - \WW \WW^T)\|_F \\
& = \| (\AA_1 - \AA_1\WW \WW^T) + (\AA_2 - \AA_2\WW \WW^T) -2 (\BB - \BB\WW \WW^T ) \|_F\\
& = 2\| \BB( \II - \WW \WW^T)\|_F \\
& \leq 2\| \BB(\II - \MM\NN^T)\|_F \\
& \leq 2(1+\epsilon) \| \BB - \BB_{k+2}\|_F 
\end{split}
\end{equation}
Next, we bound $\| \AA - \AA_k\|_F $ in terms of $\| \BB - \BB_{k+2}\|_F$ as follows: 
\begin{equation}
\begin{split}
\| \AA - \AA_k\|_F & = \| \AA_1 + \AA_2 -2\BB - \AA_k\|_F \\
& = 2\Big\| \BB - \frac{(\AA_1 + \AA_2 - \AA_k)}{2}\Big\|_F
\end{split}
\end{equation}
Observe, $\AA_1 + \AA_2 - \AA_k$ is a rank-$2$ perturbation to a rank-$k$ matrix, therefore, $\Big\| \BB - \frac{(\AA_1 + \AA_2 - \AA_k)}{2}\Big\|_F \geq \| \BB - \BB_{k+2}\|_F$. Combining the two equations, we get $\| \AA(\II - \WW \WW^T)\|_F \leq (1+\epsilon)\| \AA - \AA_k\|_F$. Recall, $\WW = [\mathbf{V}; \mathbf{a}_1, \mathbf{a}_2]$, where $\mathbf{V}$ has rank $k+2$ and thus $\WW$ has rank at most $k+4$. 
\end{proof}

\newpage

%% file: appendix.tex
\appendix
\section{Appendix.}

\begin{lemma}
Let $\AA \in \mathbb{R}^{m \times n}$ be a distance matrix. For $j \in [n]$, let $\widetilde{X}_j$ be a $O\left( \frac{m}{b} \right)$-approximate estimate for the $j^{th}$ column of $A$ such that it satisfies the guarantee of Corollary \ref{cor:column_norm_estimation}. Then, let $q = \{q_1, q_2 \ldots q_n\}$ be a probability distribution over the columns of $A$ such that $q_j = \frac{\widetilde{X}_j}{\sum_{j'}\widetilde{X}_{j'}}$. Let $t = O\left( \frac{ m k^2}{b\epsilon^2}\log(\frac{m}{\delta})\right)$ for some constant $c$. Construct $\CC$ using $t$ columns of $A$ and set each one to $\frac{A_{*,j}}{\sqrt{t q_j}}$ with probability $q_j$. Let $\ell$ be the index of the smallest singular value of $\AA$ such that $\sigma^2_\ell \geq \frac{\| \AA \|^2_F}{k}$. With probability at least $1-\delta$, 
\begin{equation*}
    \trace{\AA_{\setminus \ell} \AA_{\setminus \ell}^T} - \trace{\P_{\setminus \ell} \CC \CC^T \P_{\setminus \ell}} = \pm \epsilon \| \AA \|^2_F
\end{equation*}
\end{lemma}
\begin{proof}
We can rewrite the above equation as $ \|\P_{\setminus \ell} \CC\|^2_F - \| \AA_{\setminus \ell}\|^2_F = \pm \epsilon\| \|^2_F$. By summing over the column norms of $\CC$, we get $\|\P_{\setminus \ell} \CC\|^2_F = \sum_{j=[t]} \|\P_{\setminus \ell} \CC_{*,j} \|^2_2 $. Next, we upper bound each term in the sum as follows : 
\begin{equation}
\begin{split}
\|\P_{\setminus \ell} \CC_{*,j} \|^2_2  & = \frac{1}{t q_j} \|\P_{\setminus \ell} \AA_{*,j} \|^2_2 \\
& = \left(\frac{b \epsilon^2}{m k^2 \log(m/\delta)}\right) \left(\frac{ m\|\AA \|^2_F}{ b\|\AA_{*,j} \|^2_2} \right) \|\P_{\setminus \ell} \AA_{*,j} \|^2_2 \\
& \leq \left(\frac{ \epsilon^2}{ k^2 \log(m/\delta)}\right) \|\AA \|^2_F 
\end{split}
\end{equation}
Note, $\frac{ k^2 \log(m/\delta)}{\epsilon^2 \|\AA \|^2_F} \|\P_{\setminus \ell} \CC_{*,j} \|^2_2  \in [0, 1]$ and $\expecf{}{\sum_{j=[t]} \|\P_{\setminus \ell} \CC_{*,j} \|^2_2 } = \|\AA_{\setminus \ell}\|^2_F$. By Chernoff, 
\begin{equation}
\begin{split}
& \prob{}{\|\P_{\setminus \ell} \CC \|^2_F \geq \|\P_{\setminus \ell} \AA \|^2_F + \epsilon\|\AA\|_F} \\ & = \prob{}{\left(\frac{ k^2 \log(m/\delta)}{\epsilon^2 \|\AA \|^2_F}\right) \sum_{j\in [t]}\|\P_{\setminus \ell} \CC_{*,j} \|^2_F \geq \left(\frac{ k^2 \log(m/\delta)}{\epsilon^2 \|\AA \|^2_F}\right) \|\AA_{\setminus \ell} \|^2_F + \epsilon\|\AA\|_F} \\
& = \prob{}{\left(\frac{ k^2 \log(m/\delta)}{\epsilon^2 \|\AA \|^2_F}\right) \sum_{j\in [t]}\|\P_{\setminus \ell} \CC_{*,j} \|^2_F \geq 1 + \frac{\epsilon\|\AA\|_F}{\|\AA_{\setminus \ell} \|^2_F} \left(\frac{ k^2 \log(m/\delta)\|\AA_{\setminus \ell} \|^2_F}{\epsilon^2 \|\AA \|^2_F}\right)  }  \\
& \leq e^{-\frac{c\log(m/\delta)}{4}} \leq \delta/2
\end{split}
\end{equation}
Therefore, with probability at least $1 -\delta/2$, $\|\P_{\setminus \ell} \CC \|^2_F - \| \AA_{\setminus \ell} \|^2_F \leq \epsilon\|\AA\|_F$. Similarly, we can show that with probability at least $1 - \delta/2$, $\|\P_{\setminus \ell} \CC \|^2_F - \| \AA_{\setminus \ell} \|^2_F \geq -\epsilon\|\AA\|_F$. Union bounding over the two events finishes the proof. 
\end{proof}

\begin{lemma}(Matrix Bernstein Inequality \cite{tropp2012user}.)
\label{lem:matrix_bern}
Let $\X_1, \X_2, \ldots \X_n$ be independent random matrices with common dimension $d_1 \times d_2$ such that for $k\in [n]$, $\expec{}{\X_k} = 0$ and $\|\X_k \|_2 \leq L$, where $\|\MM\|_2$ for any matrix $\MM$, represents the operator norm. Let $\X = \sum_{k\in [n]} \X_k$. Let $\var{\X}$ denote the matrix variance statistic, i.e. $\var{\X} = \max\left( \expec{}{\X\X^T}, \expec{}{\X^T\X} \right)$. Then, for all $t\geq 0$,
\[
\Pr\left[\|\X \|_2 \geq t \right] \leq (d_1 + d_2) \exp{\left(\frac{-t^2}{\var{\X} + Lt/3}\right)}
\]
\end{lemma}

%% file: sublinear_lra.bbl
\newcommand{\etalchar}[1]{$^{#1}$}
\begin{thebibliography}{DGM{\etalchar{+}}09}

\bibitem[BDN15]{BDN15}
Jean Bourgain, Sjoerd Dirksen, and Jelani Nelson.
\newblock Toward a unified theory of sparse dimensionality reduction in
  euclidean space.
\newblock In {\em Proceedings of the Forty-Seventh Annual {ACM} on Symposium on
  Theory of Computing, {STOC} 2015, Portland, OR, USA, June 14-17, 2015}, pages
  499--508, 2015.

\bibitem[CEM{\etalchar{+}}15]{cohen2015dimensionality}
Michael~B Cohen, Sam Elder, Cameron Musco, Christopher Musco, and Madalina
  Persu.
\newblock Dimensionality reduction for k-means clustering and low rank
  approximation.
\newblock In {\em Proceedings of the forty-seventh annual ACM symposium on
  Theory of computing}, pages 163--172. ACM, 2015.

\bibitem[CMM17]{cmm17}
Michael~B. Cohen, Cameron Musco, and Christopher Musco.
\newblock Input sparsity time low-rank approximation via ridge leverage score
  sampling.
\newblock In {\em Proceedings of the Twenty-Eighth Annual {ACM-SIAM} Symposium
  on Discrete Algorithms, {SODA} 2017, Barcelona, Spain, Hotel Porta Fira,
  January 16-19}, pages 1758--1777, 2017.

\bibitem[Coh16]{c16}
Michael~B. Cohen.
\newblock Nearly tight oblivious subspace embeddings by trace inequalities.
\newblock In {\em Proceedings of the Twenty-Seventh Annual {ACM-SIAM} Symposium
  on Discrete Algorithms, {SODA} 2016, Arlington, VA, USA, January 10-12,
  2016}, pages 278--287, 2016.

\bibitem[CW13]{clarkson2013low}
Kenneth~L Clarkson and David~P Woodruff.
\newblock Low rank approximation and regression in input sparsity time.
\newblock In {\em Proceedings of the forty-fifth annual ACM symposium on Theory
  of computing}, pages 81--90. ACM, 2013.

\bibitem[DGM{\etalchar{+}}09]{DHMRTTWW09}
Erik~D. Demaine, Francisco Gomez{-}Martin, Henk Meijer, David Rappaport, Perouz
  Taslakian, Godfried~T. Toussaint, Terry Winograd, and David~R. Wood.
\newblock The distance geometry of music.
\newblock {\em Comput. Geom.}, 42(5):429--454, 2009.

\bibitem[DMM08]{drineas2008relative}
Petros Drineas, Michael~W Mahoney, and S~Muthukrishnan.
\newblock Relative-error cur matrix decompositions.
\newblock {\em SIAM Journal on Matrix Analysis and Applications},
  30(2):844--881, 2008.

\bibitem[DPRV15]{d15}
Ivan Dokmanic, Reza Parhizkar, Juri Ranieri, and Martin Vetterli.
\newblock Euclidean distance matrices: essential theory, algorithms, and
  applications.
\newblock {\em IEEE Signal Processing Magazine}, 32(6):12--30, 2015.

\bibitem[FKV04]{fkv04}
Alan~M. Frieze, Ravi Kannan, and Santosh Vempala.
\newblock Fast monte-carlo algorithms for finding low-rank approximations.
\newblock {\em J. {ACM}}, 51(6):1025--1041, 2004.

\bibitem[JS04]{js04}
Viren Jain and L.K. Saul.
\newblock Exploratory analysis and visualization of speech and music by locally
  linear embedding.
\newblock In {\em Departmental Papers (CIS)}, pages iii -- 984, 06 2004.

\bibitem[MM13]{mm13}
Xiangrui Meng and Michael~W. Mahoney.
\newblock Low-distortion subspace embeddings in input-sparsity time and
  applications to robust linear regression.
\newblock In {\em Symposium on Theory of Computing Conference, STOC'13, Palo
  Alto, CA, USA, June 1-4, 2013}, pages 91--100, 2013.

\bibitem[MW17]{mw17}
Cameron Musco and David~P. Woodruff.
\newblock Sublinear time low-rank approximation of positive semidefinite
  matrices.
\newblock In {\em 58th {IEEE} Annual Symposium on Foundations of Computer
  Science, {FOCS} 2017, Berkeley, CA, USA, October 15-17, 2017}, pages
  672--683, 2017.

\bibitem[NN13]{NN13}
Jelani Nelson and Huy~L. Nguyen.
\newblock {OSNAP:} faster numerical linear algebra algorithms via sparser
  subspace embeddings.
\newblock In {\em 54th Annual {IEEE} Symposium on Foundations of Computer
  Science, {FOCS} 2013, 26-29 October, 2013, Berkeley, CA, {USA}}, pages
  117--126, 2013.

\bibitem[Sar06]{sarlos2006improved}
Tamas Sarlos.
\newblock Improved approximation algorithms for large matrices via random
  projections.
\newblock In {\em FOCS}, pages 143--152, 2006.

\bibitem[TDSL00]{t00}
Joshua~B Tenenbaum, Vin De~Silva, and John~C Langford.
\newblock A global geometric framework for nonlinear dimensionality reduction.
\newblock {\em science}, 290(5500):2319--2323, 2000.

\bibitem[Tro12]{tropp2012user}
Joel~A Tropp.
\newblock User-friendly tail bounds for sums of random matrices.
\newblock {\em Foundations of computational mathematics}, 12(4):389--434, 2012.

\bibitem[WS04]{ws04}
Kilian~Q. Weinberger and Lawrence~K. Saul.
\newblock Unsupervised learning of image manifolds by semidefinite programming.
\newblock In {\em 2004 {IEEE} Computer Society Conference on Computer Vision
  and Pattern Recognition {(CVPR} 2004), with CD-ROM, 27 June - 2 July 2004,
  Washington, DC, {USA}}, pages 988--995, 2004.

\end{thebibliography}
